\documentclass[12pt]{amsart}
\usepackage{amsmath}
\usepackage{amssymb}
\usepackage{amsthm}
\usepackage{braket}
\usepackage[backend=biber]{biblatex}
\usepackage{hyperref}
\usepackage{mathtools}
\usepackage{mathrsfs}
\theoremstyle{plain}
\newtheorem{theorem}{Theorem}

\newtheorem{proposition}[theorem]{Proposition}
\newtheorem{corollary}[theorem]{Corollary}

\theoremstyle{definition}
\newtheorem{definition}[theorem]{Definition}

\theoremstyle{remark}

\newtheorem{remark}[theorem]{Remark}
\newtheorem{example}[theorem]{Example}

\newcommand\numberthis{\addtocounter{equation}{1}\tag{\theequation}}

\makeatletter
\newcommand{\reqnomode}{\tagsleft@false}
\makeatother

\allowdisplaybreaks

\title{A unified approach to quantum de Finetti theorems and SoS rounding via geometric quantization}
\author{Sujit Rao}
\begin{document}
\reqnomode

\begin{abstract}
The sum-of-squares hierarchy of semidefinite programs has become a common tool for algorithm design in theoretical computer science, including problems in quantum information. In this work we study a connection between a Hermitian version of the SoS hierarchy, related to the quantum de Finetti theorem, and geometric quantization of compact Kähler manifolds (such as complex projective space $\mathbb{C}P^{d}$, the set of all pure states in a $(d + 1)$-dimensional Hilbert space). We show that previously known HSoS rounding algorithms can be recast as quantizing an objective function to obtain a finite-dimensional matrix, finding its top eigenvector, and then (possibly nonconstructively) rounding it by using a version of the Husimi quasiprobability distribution. Dually, we recover most known quantum de Finetti theorems by doing the same steps in the reverse order: a quantum state is first approximated by its Husimi distribution, and then quantized to obtain a separable state approximating the original one. In cases when there is a transitive group action on the manifold we give some new proofs of existing de Finetti theorems, as well as some applications including a new version of Renner's exponential de Finetti theorem proven using the Borel--Weil--Bott theorem, and hardness of approximation results and optimal degree-2 integrality gaps for the basic SDP relaxation of \textsc{Quantum Max-$d$-Cut} (for arbitrary $d$). We also describe how versions of these results can be proven when there is no transitive group action. In these cases we can deduce some error bounds for the HSoS hierarchy on complex projective varieties which are smooth.
\end{abstract}

\maketitle

\section{Introduction}
Quantum mechanics and quantum information are often thought of as being similar to classical probability theory, but the analogy is of course not exact and must break down at some point. The point where the analogy fails is often illustrated through the use of quasiprobability distributions, which associate a classical ``distribution'' (over phase space) to a quantum state in a way which preserves as many properties as possible. In the semiclassical limit $\hbar \to 0$, the mapping must preserve essentially all properties of the state. A popular choice is the Wigner quasiprobability distribution, which preserves many properties of the original quantum state, but may take on negative values and is thus not a classical probability distribution.

Elsewhere in computer science, negative probabilities have arisen in optimization. One technique for approximating solutions to non-convex problems involves relaxing a problem of the form $\max _{x \in D} p(x)$ for some domain $D$ and objective function $p$ into a convex problem of the form $\max _{\mu \in \mathcal{P}(D)} E_{\mu}[p(x)]$, where $\mathcal{P}(D)$ is set of probability distributions over $D$. Since this new convex problem often has dimension which is exponentially or infinitely large, the strategy of sum-of-squares (SoS) optimization is to further enlarge the set of distributions to \textit{degree-$k$ pseudo-distributions}, which have an efficient finite-dimensional representation. Such pseudo-distributions also have an interpretation in terms of negative probabilities: a degree-$k$ pseudo-distribution can be defined as a density function $P : D \to \mathbb{R}$ (relative to some base measure $\mu$) such that $\int _{D} q(x)P(x)\,d\mu(x) \geq 0$ for all degree-$k$ polynomials $q$ which are sums-of-squares. As the degree $k \to \infty$, the value of the relaxed problem approaches that of the original one.

It is natural to ask whether there is a connection between these two notions of negative probability, where where the inverse $1/\hbar$ of the semiclassical parameter is analogous to the SoS degree $k$. In this work, we propose an affirmative answer to this question.

\subsection{Approaches to nonlinear classical optimization}
For solving purely a classical problem which is nonlinear, such as optimization of a nonlinear polynomial function over a compact semialgebraic set, a common approach involves first constructing a finite-dimensional linear approximation to the problem, solving it, and then using the solution to approximate the original nonlinear problem.
Indeed, the sum-of-squares optimization hierarchy does exactly this, where the finite-dimensional approximation is a semidefinite program with a linear objective function.
Another common approach is to construct a finite set of reasonably separated points on which the objective function can be evaluated, and then approximating the optimum on the entire space by the optimum on the finite set.
In either case, the approximation to the nonlinear problem must be finite in some way -- being finite-dimensional in the case of linearization, or having finite cardinality in the case of brute-forcing over a finite subset.

Motivated by physics, one general approach to approximating a nonlinear function by a linear operator is \textit{quantization}.
In quantization one thinks of a nonlinear function as being a classical observable on some phase space, and the corresponding linear operator is the corresponding quantum observable acting on a Hilbert space, which should again correspond to the classical phase space in some way.
The correspondence should depend on a parameter $\hbar$, and in the semiclassical limit $\hbar \to 0$ the correspondence should become closer and closer.

This approach suggests that there should be corresponding algorithms for various computational tasks, but for optimizing a nonlinear polynomial function in particular there is a clear analogy:
one would first quantize the objective function, and then approximate its minimum or maximum by the minimum or maximum of the corresponding observable.
One could even try to approximate the optimal solution itself by finding a classical distribution over the phase space which approximates the top eigenstate of the observable, and then sampling from it.
One of the most common such \textit{quasiprobability distributions} is the Husimi Q-function, which, unlike other common quasiprobability distributions, such as the Wigner function or Glauber-Sudarshan P-function, is nonnegative and integrates to 1, giving a well-defined classical probability distribution. 
The parameter $\hbar$ then allows one to tune the approximation, for which a smaller value of $\hbar$ would give a more accurate approximation but require more resources computationally.

Conversely, one could also try to approximate a quantum state by a classical objective. Without additional assumptions one typically expects a quantum state on a large number of sites to have nontrivial entanglement. However, under symmetry conditions the principle of \textit{monogamy of entanglement} suggests that no two sites should be strongly entangled with each other, implying that the overall state behaves classically. Indeed, the quantum de Finetti theorem \cite{Christandl_Koenig_Mitchison_Renner_2007} formalizes this intuition with a quantitative error bound. The construction in the de Finetti theorem is quite similar to the rounding algorithm described above -- one can think of the Husimi quasiprobability distribution as a mapping from a quantum mixed state to a classical mixed state, and the Glauber-Sudharshan P-quantization as a mapping in the opposite direction, such that the two maps are approximate inverses.

\subsection{Results in this paper}
As we will see, it is possible to make sense of some version of this quantization-based optimization algorithm. It will turn out to be exactly the same as a Hermitian version of the sum-of-squares hierarchy, but several assumptions will need to be made on the original problem.
First, the domain $M$ of the optimization problem should be a compact semialgebraic set in order for the original optimization problem to make sense and for a Hermitian SoS relaxation to even be defined. To be able to think of $M$ as a classical phase space we will also need $M$ to be a smooth manifold, and have a Poisson bracket defined on pairs of classical observables, which is the minimum structure needed to define time-evolution of a system in classical mechanics. Because of the quantization procedure we use, we will actually need to further assume that the Poisson bracket comes from a symplectic structure on $M$ and that $M$ is also a polarized K\"{a}hler manifold (comes with a holomorphic line bundle $\mathcal{L} \to M$ which defines an embedding into a projective space).

\subsection{Content of this paper}
In Section~\ref{sec:tech-overview}, we give a more detailed technical overview and give a more precise definition of quantization. Then in Section~\ref{sec:warmup} we describe one of the simplest examples of quantization, which is a single bosonic mode with phase space $\mathbb{C} \cong \mathbb{R}^{2}$. (Despite the Hilbert space being infinite-dimensional in this case, it is overall simpler because there is no topology involved.) In Section~\ref{sec:geom-quant} we describe \textit{geometric quantization}, which is one standard approach to quantizing a compact phase space, which in general must be topologically nontrivial. In Section~\ref{sec:main-sec} we state and prove Theorem~\ref{thm:main-thm}, which is a version of the de Finetti theorem generalized to the geometric quantization framework. Finally in Section~\ref{sec:applications} we show how to deduce various known quantum de Finetti theorems, including a version of Renner's exponential de Finetti theorem and the de Finetti theorem for irreducible representations of $U(n)$, from our general theorem, and how to deduce some new de Finetti theorems for other phase spaces, including smooth projective varities; in these cases we need to apply some recent results in geoemtric analysis and complex geometry to show that the assumptions of Theorem~\ref{thm:main-thm} can be satisfied, and also to explain how they can be used to deduce quantitative error bounds.

\section{Technical overview}\label{sec:tech-overview}
\subsection{Hermitian sum-of-squares hierarchies}
Consider the problem of optimizing a polynomial function $p : M \to \mathbb{R}$, where $M$ is a compact manifold (with additional structure to be specified later). One hypothetical class of algorithms could take the following form:
\begin{enumerate}
\item Treating the manifold $M$ as a classical phase space, choose some $\hbar$ and quantize $M$ to obtain a finite-dimensional Hilbert space $\mathcal{H}$. Then quantize $p$ to obtain a self-adjoint observable $Q(p) : \mathcal{H} \to \mathcal{H}$.
\item Calculate the minimum or maximum eigenvalue of $Q(p)$ and corresponding eigenvector $\psi \in \mathcal{H}$.
\item Round $\psi$ to a point $x \in M$ by sampling from a probability distribution over $M$ whose density is given by the Husimi Q-function of $\psi$ (which is nonnegative and integrates to 1). Actually sampling from the corresponding probability distribution could be nontrivial and may only give a nonconstructive approximation bound, as opposed to a full rounding algorithm.
\end{enumerate}

Such an algorithm was implicit in a paper of Quillen \cite{Quillen_1968}. In Quillen's result, we take $M = \mathbb{C}P^{d}$ (complex projective space, the manifold of pure states in a $d$-dimensional Hilbert space), the quantization $Q$ to be an analog of the Husimi quantization rule, and the Husimi Q-function to be an analog of the usual bosonic Husimi Q-function.
The problem of calculating the extremal eigenvalue of $Q(p)$ turns out in this case to be the same as the SDP coming from a degree-$k$ Hermitian sum-of-squares relaxation, with $\hbar = 1/k$.
Quillen's paper only proves an asymptotic result for the minimization problem (equivalently, showing that a nonnegative Hermitian polynomial is a Hermitian sum-of-squares), but the techniques use quantization formulas involving creation and annihilation operators and are very similar to the proofs in \cite[Theorem~2.2]{Lewin_Nam_Rougerie_2015} and \cite[Theorem~4.4]{Rougerie_2020} of the quantum de Finetti theorem using Chiribella's formula \cite[Equation~(6)]{Chiribella_2011}.

\subsubsection{Analysis of the approximation and duality between quantization rules and quasiprobability distributions}
A quasiprobability distribution defines a linear map $P : B(\mathcal{H}) \to \mathcal{P}(M)$ from density matrices to signed measures of total mass 1 on the classical phase space $M$. Assuming $M$ comes with a canonical measure (for example, if it has a canonical Riemannian metric), then we can try to take the transpose or adjoint of $P$ to obtain a quantization rule $Q : C^{\infty}(M) \to B(\mathcal{H})$. The fact that $Q$ is the adjoint of $P$ means that it satisfies the property
\[ \langle f, P(\rho) \rangle = \langle Q(f), \rho \rangle. \]
On the left-hand side $\langle \cdot, \cdot \rangle$ denotes the expectation value of a classical function with respect to a measure of total mass 1, and on the right-hand side it denotes the expectation of a quantum observable with respect to a mixed state or density matrix.

The analysis requires calculating the expectation value of a classical function $f$ with the Husimi Q-function, which requires quantizing the classical function with respect to the corresponding quantization rule. This quantization rule will end up being the Glauber-Sudarshan P-quantization rule. Since the original relaxation was defined using the Husimi Q-quantization rule, the analysis depends on analyzing the difference between the quantization of $f$ using the P-rule and the quantization of $f$ with respect to the Q-rule. General intuition from quantum mechanics predicts that the difference is $O(\hbar)$ (ignoring constants depending on $M$ and $f$), which is the same as $O(1/k)$. Indeed, we will see that a mathematically rigorous analysis recovers the same error bound and can also determine the dependence on $M$ and $f$.

\subsubsection{Compactness of the classical phase space and finite-dimensionality of the quantum Hilbert space}\label{sec:compact-vol-hbar}
Note that the classical phase space is compact and the Hilbert space $\mathcal{H}$ is finite-dimensional, The intuition for this might be a little bit unclear, but a standard explanation is that the dimension of the Hilbert space corresponds to the total number of ``cells'' of size $\hbar$ needed to cover $M$. If the classical phase space is compact and has finite volume, the Hilbert space will then be finite-dimensional. As $\hbar \to 0$ we have $\dim \mathcal{H} \to \infty$. This matches the sum-of-squares perspective, where $1/k = \hbar \to 0$ is equivalent to the SoS degree $k \to \infty$ and the size of the SDP correspondingly increases.

An alternative explanation starts from the quantum dynamics with a finite-dimensional Hilbert space, assuming that the operators satisfy certain commutation relations. Using a spin-$s$ Hilbert space, if we define $\hbar = 1/s$ then after doing a rescaling the spin operators satisfy the relations
\begin{align*}
1 &= x^{2} + y^{2} + z^{2} \\
[x, y] &= i\hbar z \\
[y, z] &= i\hbar x \\
[z, x] &= i\hbar y.
\end{align*}
Thus a classical description of the dynamics in the limit $\hbar \to 0$ should label a state by a point $(x, y, z) \in S^{2}$ in the sphere, and the algebra of classical observables on $S^{2}$ should have a Poisson bracket matching the quantum commutators to first order in $\hbar$. In particular, the phase space in this case should be $S^{2}$, which is compact.

\section{Warm-up: quantizing a single bosonic mode}\label{sec:warmup}
To prepare for describing the quantization of a compact phase-space, we first describe the quantization of the standard phase space $\mathbb{R}^{2} \cong \mathbb{C}$ for a single bosonic modes. At first this may seem more technical since the resulting Hilbert space is infinite-dimensional, but in some ways it is simpler: a compact phase space must be topologically nontrivial (since the symplectic form defines a nonzero de Rham cohomology class), so one must consider sections of a line bundle instead of holomorphic functions. For a non-compact phase space which is contractible (is homotopy equivalent to a point) there are enough holomorphic functions to define the Segal-Bargmann space, which is isomorphic to the usual single-mode Hilbert space. A model of continuous-variable quantum computation, also represented using the Segal-Bargmann space, was proposed in \cite{Chabaud_Mehraban_2022}.

\subsection{Classical phase space}
The phase space for a single classical bosonic mode is $M = \mathbb{R}^{2}$ with position coordinate $x$ and momentum coordinate $p$, and symplectic form $\mathrm{d}x\wedge \mathrm{d}p$. The symplectic form allows one to define the Poisson bracket $\{f, g\}$ of two classical observables $f$ and $g$, which is a classical analog of the quantum commutator. The Poisson bracket satisfies the properties
\begin{align*}
\{x, p\} &= 1 \\
\{g, f\} &= -\{f, g\} \\
\{f, gh\} &= \{f, g\}h + \{f, h\}g.
\end{align*}

In quantization, the goal is typically to construct a family of Hilbert spaces $\mathcal{H}_{\hbar}$ depending on a parameter $\hbar$, along with a linear mapping from classical observables $f : M \to \mathbb{R}$ (real-valued functions on the phase space) to quantum observables $Q(f) : \mathcal{H}_{\hbar} \to \mathcal{H}_{\hbar}$ (self-adjoint operators on the Hilbert space). The quantization map $Q$ is usually also required to satisfy that the quantum commutator agrees with the classical Poisson bracket to first-order in $\hbar$. Explicitly, we want to have $[Q(f), Q(g)] = \hbar \{f, g\} + O(\hbar^{2})$.

\subsection{Position-space quantization}
For a bosonic mode, there is a standard construction of a quantization arising from early work on quantum mechanics. The Hilbert space is taken to be
\[ \mathcal{H}_{\hbar} \coloneqq L^{2}(\mathbb{R}) = \left\{\text{$f : \mathbb{R} \to \mathbb{R}$ measurable} \,\middle|\, \int _{\mathbb{R}} |f|^{2} < \infty\right\} \]
which is the space of normalized wavefunctions in position space. To quantize $x$ and $p$, a typical definition is
\begin{align*}
Q(x) &= (f \mapsto x\cdot f) \\
Q(p) &= (f \mapsto i\hbar (df/dx)).
\end{align*}
To quantize more complicated observables, there are multiple different choices which are used in different contexts. Concretely, if a classical observable $f$ is a polynomial in $x$ and $p$ one can try to define $Q(f)$ as a particular noncommutative polynomial in $Q(x)$ and $Q(p)$, but there will be multiple noncommutative polynomials which reproduce $f$ in the $\hbar \to 0$ limit (which corresponds to formally forcing $Q(x)$ and $Q(p)$ to commute). Some common choices include always taking $Q(x)$ to be on the left and $Q(p)$ on the right (the Mehta prescription), $Q(p)$ always to the left and $Q(x)$ always on the right (the Kirkwood--Rihaczek prescription), or taking a uniform linear combination over all $d!$ orderings for a polynomial of degree $d$ (the Weyl prescription, which is the most well-known and commonly used). In terms of creation and annihilation operators, one can also consider rules where all creation operators are to the left and annihilation operators are to the right (called Wick or normal ordering, and adjoint to the Husimi Q-function), all creation operators are to the right and all annihilation operators are to the left (called anti-Wick or anti-normal ordering, and adjoint to the Glauber-Sudarshan P-function), or a uniform linear combination over all $d!$ orderings (which is again the Weyl prescription). More explicit formulas for these quantization rules can be found in \cite[Section~0.19]{Curtright_Fairlie_Zachos_2013}

\subsection{Holomorphic quantization}
When quantizing, one might expect that the wavefunctions should be defined over the entire phase space. However, as seen from the position Hilbert space the wavefunctions in general should depend only on a single coordinate. One can alternatively take momentum Hilbert space, where the wavefunctions depend on a momentum coordinate. Under a Fourier transform, the momentum Hilbert space defines a representation of the position and momentum operators which is isomorphic to the position Hilbert space.

In general forms of geometric quantization, the wave functions are defined on the entire phase space and then required to be constant with respect to a ``polarization.'' For a bosonic mode there are polarizations which recover the position and momentum Hilbert spaces, as well as a continuum of polarizations interpolating between them. In this work we will only use the ``holomorphic polarization,'' which in a sense is the one which is exactly half-way in between. Explicitly, the holomorphic Hilbert space for a bosonic mode is
\[
L^{2}_{\text{hol},\hbar}(\mathbb{C}) \coloneqq \left\{ \text{$f : \mathbb{C} \to \mathbb{C}$ holomorphic} \, \middle| \, \frac{1}{\pi}\int _{\mathbb{C}} |f(z)|^{2}e^{-\hbar|z|^{2}}\,dz < \infty \right\}.
\]

Instead of defining the position and momentum operators, it is slightly simpler to define the creation and annihilation operators as
\begin{align*}
Q\left(\frac{x - ip}{\sqrt{2}}\right) &= (f \mapsto zf) \\
Q\left(\frac{x + ip}{\sqrt{2}}\right) &= (f \mapsto df/dz)
\end{align*}
where $df/dz$ is the Wirtinger derivative of $f$. One can check that $\{z^{n}/\sqrt{n!}\}$ is an orthonormal basis of $L^{2}_{\text{hol}}$. An explicit isomorphism, known as the \textit{Segal-Bargmann transform}, from the position Hilbert space to $L^{2}_{\text{hol}}$ maps the $n$-particle number states as $\ket{n} \mapsto z^{n}/\sqrt{n!}$. The Segal-Bargmann transform also intertwines the creation and annihilation operators on the holomorphic Hilbert space with those on the position Hilbert space.

\subsection{Coherent states}
The usual position and momentum operators are linear combinations of the creation and annihilation operators, so all the quantization rules described above can be defined on the holomorphic Hilbert space as well. The standard coherent states of a harmonic oscillator correspond to the holomorphic functions
\[ \ket{\psi_{a}} = (z \mapsto \exp(-|a|^{2})\exp(\bar{a} z)). \]
This follows from taking the harmonic oscillator ground state, which corresponds to the holomorphic function $1$, and then applying a displacement operator to it. The Husimi Q-function turns out to be
\[ Q_{f}(z) = |\braket{\psi_{a}|f}|^{2} = \frac{1}{\pi}|f(z)|^{2}\exp(-|z|^{2}). \]
In fact, the more precise statement
\[ \braket{\psi_{a}|f} = \exp(-|a|^{2})f(z) \]
is also true. In the following sections, we will take analogs of these formulas as characterizations of the Husimi Q-function and coherent states for compact phase spaces.

\section{Quantizing a compact phase space}\label{sec:geom-quant}
Fix a compact K\"{a}hler manifold $(M, g, \omega, J)$, where $g$ is the Riemannian metric, $\omega$ is the symplectic form, and $J$ is the complex structure. The manifold $M$ will serve as the phase space of the classical mechanical system and the domain of a classical polynomial optimization problem. Furthermore, fix some arbitrary Borel measure $\mu$ on $M$, which does not necessarily have to be the volume measure coming from the K\"{a}hler structure. This additional generality will make it easier to bound the error in our generalized quantum de Finetti theorem, and some intuition is given in the beginning of Section~\ref{sec:main-sec}.

\begin{example}
A running example used throughout will be $M = \mathbb{C}P^{1}$, with its usual round Riemannian metric, standard complex structure, and Fubini-Study symplectic form. This is the same as the Riemann sphere in complex analysis, or the Bloch sphere $S^{2}$ in quantum information. It is also the compact phase space corresponding to spin representations of the angular momentum Lie algebra $\mathfrak{su}(2)$.

Explicitly, the K\"{a}hler structure can be defined in coordinates. The standard atlas on $\mathbb{C}P^{1}$ has two charts, whose domains are $\mathbb{C}P^{1} \setminus \{\infty\} \cong \mathbb{C}$ and $\mathbb{C}P^{1} \setminus \{0\}$, thinking of $\mathbb{C}P^{1}$ as the Riemann sphere. (In the Bloch sphere, these can be identified with the $\pm 1$ eigenstates of the Pauli $Z$ operator.)
The intersection of the domains can be identified with $\mathbb{C} \setminus \{0\}$, and the transition map between the two charts is $z \mapsto z^{-1}$. (This is an involution and so it is in fact the transition map in either direction.) The tangent space has 1 complex dimension, and the corresponding Hermitian metric is the $1$-by-$1$ matrix $(1 + |z|^{2})^{-1}$ in the chart whose domain is $\mathbb{C}P^{1} \setminus \{\infty\}$. It can be checked that this metric is K\"{a}hler.

Alternatively, one can use symplectic reduction to define the symplectic form, and then recover the other parts of the K\"{a}hler structure. In this case we would start with the phase space $\mathbb{C}^{2}$, with the $U(1)$ action given by $z\cdot v = zv$ for $z \in U(1)$ and $v \in \mathbb{C}^{2}$. The corresponding momentum map is $\frac{1}{2}|v|^{2}$, and we first take the submanifold
\[ \{v \in \mathbb{C}^{2} : |v|^{2} = 1\} \cong S^{3}. \]
The quotient by $U(1)$ is then $\mathbb{C}P^{1} \cong S^{2}$, which can be identified with the Riemann sphere or the Bloch sphere.
\end{example}

Further, fix a holomorphic line bundle $\mathcal{L} \to M$ with a Hermitian metric $h_{\mathcal{L}}$.

\begin{definition}[Geometric quantization]
The quantum Hilbert space associated to the pair $(M, \mathcal{L})$ is the space $\mathcal{H}_{\mathcal{L}} \coloneqq \Gamma_{\text{hol}}(\mathcal{L})$ of holomorphic global sections of the bundle $\mathcal{L}$. The inner product on $\mathcal{H}_{\mathcal{L}}$ is defined by
\[ \langle s_{1}, s_{2} \rangle_{\mathcal{H}_{\mathcal{L}}} \coloneqq \int _{M} h_{\mathcal{L}}(s_{1}(x), s_{2}(x))\,d\mu(x) \]
where $\mu$ is the measure on $M$ from before.
\end{definition}

For a fixed line bundle $\mathcal{L}$ it is common to consider tensor powers $\mathcal{L}^{\otimes k}$ with the tensor product Riemannian metric in the geometric quantization construction. Intuitively, the value of $k$ can correspond either to a ``total energy'' quantity, as explained in the example below, or to $k = 1/\hbar$ as described in Section~\ref{sec:compact-vol-hbar}.

\begin{example}
The spin phase space $\mathbb{C}P^{1}$ has a standard holomorphic line bundle $\mathcal{L}$ called the \textbf{canonical line bundle}. The total space of the bundle can be identified with
\[ \{ (x, v) : x \in \mathbb{C}P^{1}, v \in \mathbb{C}x \subseteq \mathbb{C}^{2} \} \]
where $x \in \mathbb{C}P^{1}$ is identified with the corresponding one-dimensional subspace of $\mathbb{C}^{2}$. There is a standard Hermitian metric on $\mathcal{L}$ whose curvature is the standard K\"{a}hler structure on $\mathbb{C}P^{1}$. If we take $\mathcal{E} = \mathbb{C}P^{1} \times \mathcal{C}$ to be the product line bundle and the line bundle to be $(\mathcal{L}^{*})^{\otimes k}$ when quantizing, then the resulting Hilbert space is
\[ \Gamma_{\text{hol}}((\mathcal{L}^{*})^{\otimes k}) \cong \mathbb{C}[z_{1}, z_{2}]_{k} \]
where $\mathbb{C}[z_{1}, z_{2}]_{k}$ is the space of homogeneous polynomials of degree $k$. This can be identified with the Hilbert space of a spin-$(k - 1)/2$ degree of freedom, since both have the same dimension. This Hilbert space is an irreducible representation of $SU(2)$, where the action of $g \in SU(2)$ is $(g \cdot p)(z_{1}, z_{2}) = p(g \cdot (z_{1}, z_{2}))$.
To complete the identification we should also specify how to quantize more general classical observables on the phase space, which will be described later.
\end{example}

\subsection{Husimi Q-function and coherent states}\label{sec:qfn-coh}
\begin{definition}
Let $s \in \mathcal{H}_{\mathcal{L}}$ be a pure state in the quantized Hilbert space. The \textbf{Husimi $Q$-function} associated to $s$ is $Q_{s} \in C^{\infty}(M)$ defined by
\[ Q_{s} \coloneqq h_{\mathcal{L}}(\overline{s} \otimes s). \]
If $\rho \in \operatorname{End}(\mathcal{H}_{\mathcal{L}})$ is a positive self-adjoint operator defining a mixed state, then it defines a section $s_{\rho} \in \Gamma_{\text{hol}}(\bar{\mathcal{L}} \boxtimes \mathcal{L})$ where $\bar{\mathcal{L}} \boxtimes \mathcal{L}$ is the external tensor product of the line bundles $\bar{\mathcal{L}} \to \bar{M}$ and $\mathcal{L} \to M$, which results of in a line bundle over the product manifold $\bar{M} \times M$. (It is known that $\Gamma_{\text{hol}}(\bar{\mathcal{L}} \boxtimes \mathcal{L}) \cong \Gamma_{\text{hol}}(\bar{\mathcal{L}}) \otimes \Gamma_{\text{hol}}(\bar{\mathcal{L}})$.)
The Husimi $Q$-function of $\rho$ is then defined by
\[
Q_{\rho} \coloneqq h_{\mathcal{L}} \circ s_{\rho} \circ \iota_{\Delta}
\]
where $\iota_{\Delta} : M \to \bar{M} \times M$ is the inclusion of the diagonal $\Delta \subseteq \bar{M} \times M$ defined by $\iota_{\Delta}(x) = (x, x)$.
One can check that this agrees with the definition for pure states, in that the same function is obtained by diagonalizing $\rho$ and then taking the corresponding convex combination of the Husimi $Q$-functions of the eigenvectors of $\rho$.
\end{definition}

The coherent states associated to a point in $x \in M$ can then be constructed from the Husimi $Q$-function. Specifically, the $Q$-function is usually thought of (when $\mathcal{E}$ is a product line bundle) as defined by the squared inner product $Q_{\rho}(x) = |\langle c_{x} | \rho | c_{x} \rangle|^{2}$ where $c_{x}$ is the coherent state associated to the point $x$. For us it will be easier to take the definition of the $Q$-function as above and then use it to define the coherent states by checking that the function $\rho \mapsto Q_{\rho}(x)$ defines a positive-definite quadratic form of rank and trace 1. The unique eigenvector of eigenvalue 1 of the quadratic form then gives the desired coherent state.

\begin{proposition}\label{prop:coherent-state}
Fix a point $x \in M$. Define the sesquilinear map $e_{\mathcal{L},x} : \operatorname{End}(\Gamma_{\text{hol}}(\mathcal{E} \otimes \mathcal{L})) \to \mathbb{C}$ by
\[ e_{\mathcal{L},x}(\rho) \coloneqq Q_{\rho}(x)\]
and the function $T_{\mathcal{L}} : M \to \mathbb{R}$ by
\[ T_{\mathcal{L}}(z) \coloneqq \sum _{i} |s_{i}(z)|_{\mathcal{L}}^{2} \]
for any orthonormal basis $\{s_{i}\}$ of $\mathcal{H}_{\mathcal{L}}$. Using the inner product on $\mathcal{H}_{\mathcal{L}}$, we identify $e_{\mathcal{L},x}$ with a linear map $\mathcal{H}_{\mathcal{L}} \to \mathcal{H}_{L}$.
Then $c_{\mathcal{L},x} \coloneqq e_{\mathcal{L},x}/T_{\mathcal{L}}(x)$ defines a pure state, in that it has the following properties:
\begin{enumerate}
\item $c_{\mathcal{L},x}$ is linear
\item $c_{\mathcal{L},x}$ is positive
\item $\operatorname{Tr} c_{\mathcal{L},x} = 1$
\item $\operatorname{rank} c_{\mathcal{L},x} = 1$.
\end{enumerate}
\end{proposition}

\begin{proof}\quad
\begin{enumerate}
\item From the definition of $Q_{\rho}$ it follows that the map $\rho \mapsto Q_{\rho}$ is linear, since $h_{\mathcal{L}}$ is bilinear. Then since $h_{\mathcal{E}}$ is bilinear it follows that $e_{v}$ is also linear.
\item Suppose that $\rho$ is a positive operator. It suffices to show that $e_{v}(\rho) \geq 0$. This follows from the fact that $Q_{\rho}$ is self-adjoint with respect to $h_{\mathcal{E},x}$ and that $h_{\mathcal{E},x}$ is a positive Hermitian inner product on $\mathcal{E}_{x}$.
\item For any orthonormal basis $\{s_{i}\}$ of $\mathcal{H}_{\mathcal{L}}$, we have
\begin{align*}
\operatorname{Tr} e_{\mathcal{L},x} &= \sum _{i} e_{\mathcal{L},x}(\bar{s}_{i} \otimes s_{i}) \\
&= \sum _{i} h_{\mathcal{L}}(\bar{s}_{i}(x) \otimes s_{i}(x)) \\
&= \sum _{i} |s_{i}(x)|_{\mathcal{L}}^{2} \\
&= T_{\mathcal{L}}(x).
\end{align*}
Thus
\[ \operatorname{Tr} c_{\mathcal{L},x} = \operatorname{Tr}(e_{\mathcal{L},x})/T_{\mathcal{L}}(x) = 1. \]
\item The map $\mathcal{H}_{\mathcal{L}} \to \mathcal{L}_{x}$ given by $s \mapsto s(x)$ goes from a finite-dimensional vector space to a 1-dimensional vector space. Since $e_{\mathcal{L},x}$ is a quadratic form given by composing this map with the quadratic form $h_{\mathcal{L}_{x}}$, the corresponding linear map must have rank 1. Then $c_{\mathcal{L},x}$ is a scalar multiple of $e_{\mathcal{L},x}$, so it also has rank 1.
\end{enumerate}
\end{proof}

\begin{proposition}\label{prop:coh-span}
The set $\{c_{\mathcal{L},x}\}$ of coherent states spans $\mathcal{H}_{\mathcal{L}}$.
\end{proposition}
\begin{proof}
We use proof by contradiction. Suppose that the coherent states span a proper subspace $\mathcal{H}' \subseteq \mathcal{H}_{\mathcal{L}}$. Let $\ket{\psi}\bra{\psi}$ be any pure state orthogonal to $\mathcal{H}'$. Then
\[ 0 = \operatorname{Tr}(c_{\mathcal{L},x} \psi) = (T_{\mathcal{L}}(x))^{-1}|\psi(x)|_{h_{\mathcal{L}}}^{2} \]
for all $x \in M$,which implies that $\psi(x) = 0$ for all $x$. Thus $\psi = 0$, which is a contradiction.
\end{proof}

The purpose of the function $T_{\mathcal{L}}$ is to normalize the coherent states to be unit vectors. To define the normalized POVM associated with the coherent states, we will need to ``un-normalize'' and put this factor back in.

\begin{definition}
The coherent state POVM associated to $M$, $\mathcal{L}$, and $\mu$ is defined to be
\[ T_{\mathcal{L}}(x) c_{\mathcal{L},x}\,d\mu(x). \]
\end{definition}

\begin{proposition}\label{prop:resolution-identity}
The coherent states define a resolution of the identity. Equivalently, the coherent state POVM is appropriately normalized, or
\[ \int_{M} T_{\mathcal{L}}(x) c_{\mathcal{L},x}\,d\mu(x) = \operatorname{id}_{\mathcal{H}_{\mathcal{L}}}. \]
\end{proposition}
\begin{proof}
Let $s \in \mathcal{H}_{\mathcal{L}}$ be any vector. It suffices to calculate
\begin{align*}
\bra{s} \left(\int_{M} T_{\mathcal{L}}(x) c_{\mathcal{L},x}\,d\mu(x)\right) \ket{s} &= \int _{M} \braket{s | e_{\mathcal{L},x} | s}\,d\mu(x) \\
&= \int _{M} |s(x)|^{2}_{\mathcal{L}_{x}}\,d\mu(x) \\
&= |s|_{\mathcal{H}_{\mathcal{L}}}^{2}.
\end{align*}
\end{proof}

\subsection{Glauber-Sudarshan P-function}
With the coherent states defined, we can define the Glauber-Sudarshan P-function of an operator.

\begin{definition}\label{def:Pquant}
Let $A \in \operatorname{End}(\mathcal{H}_{\mathcal{L}})$ be a self-adjoint operator. The \textbf{Glauber-Sudarshan P-function} of $A$ is any function $f : M \to \mathbb{R}$ such that
\[ A = \int _{M} T_{\mathcal{L}}(x)f(x)c_{\mathcal{L},x}\,d\mu(x) \]
and such that there exists some $B \in \operatorname{End}(\mathcal{H}_{\mathcal{L}})$ such that $f = Q_{B}$ is the Husimi Q-function of $B$. The \textbf{Glauber-Sudarshan P-quantization} of a classical function $f : M \to \mathbb{R}$ is defined to be
\[ \int _{M} T_{\mathcal{L}}(x)f(x)c_{\mathcal{L},x}\,d\mu(x). \]
\end{definition}

It may seem unnatural at first to require $f$ to coincide with the Husimi Q-function of some state, but this condition is actually needed for $f$ to be uniquely defined. The condition given by just the integral defines a finite number of linear constraints on $f$, since $\mathcal{H}_{\mathcal{L}}$ is finite-dimensional, while the space of smooth functions is infinite-dimensional. Thus $f$ is not uniquely defined until we constrain it to be in a finite-dimensional subspace of $C^{\infty}(M)$, which is exactly what the second condition does.

\begin{proposition}
The P-quantization defines a linear map $C^{\infty}(M) \to \operatorname{End}(\mathcal{H}_{\mathcal{L}})$ which is the adjoint of the map sending an operator to its Husimi Q-function.
\end{proposition}
\begin{proof}
We have
\begin{align*}
\left\langle B, \int _{M} \ket{c_{\mathcal{L},x}}T_{\mathcal{L}}(x)f(x)\bra{c_{\mathcal{L},x}}\,d\mu(x) \right \rangle &= \int _{M} \operatorname{Tr}(B\ket{c_{\mathcal{L},x}}T_{\mathcal{L}}(x)f(x)\bra{c_{\mathcal{L},x}})\,d\mu(x) \\
&= \int _{M} T_{\mathcal{L}}(x)\braket{c_{\mathcal{L},x} | B | c_{\mathcal{L},x}}f(x)\,d\mu(x) \\
&= \langle Q_{B}(x), f \rangle.
\end{align*}
\end{proof}

\begin{proposition}
For any $A$, the Glauber-Sudarshan P-function exists and is unique.
\end{proposition}
\begin{proof}
It suffices to show that the map $B \mapsto Q_{B}$ taking a self-adjoint operator $B$ to its Husimi Q-function is injective (so an invertible linear map onto its image), which implies that its adjoint is an invertible linear map from the image to the set of self-adjoint operators.

To this end, suppose the $Q_{B}(x) = 0$ for all $x$. Then $\braket{c_{\mathcal{L},x} | B | c_{\mathcal{L},x}} = 0$ for all $x$. The function $\braket{c_{\mathcal{L},x} | B | c_{\mathcal{L},y}}$ is holomorphic on $\bar{M} \times M$, and must be 0 because its restriction to the diagonal is 0. Since $\{c_{\mathcal{L},x}\}$ spans $\mathcal{H}_{\mathcal{L}}$ we get $B = 0$.
\end{proof}

\subsection{Geometric quantization and vector bundles}\label{sec:quant-vecbund}
In the previous section, we only considered Hilbert spaces of the form $\Gamma_{\text{hol}}(M, \mathcal{L})$ for a Hermitian line bundle $\mathcal{L}$. For certain applications, particularly to recover the exponential quantum de Finetti theorem of Renner in Section~\ref{sec:renner-exp}, we will need to consider Hilbert spaces of the form $\Gamma_{\text{hol}}(M, \mathcal{E} \otimes \mathcal{L})$ where $\mathcal{E}$ is a Hermitian vector bundle with metric $h_{\mathcal{E}}$ and $\mathcal{L}$ is a Hermitian line bundle. For a reference with slightly more detail we refer the reader to \cite{Lazarsfeld_2004}.

Our construction will essentially reduce to the case from before where $\mathcal{E}$ is trivial. Specifically, consider the unit sphere bundle $S(\mathcal{E})$ over $\mathcal{E}$. This has a $U(1)$-action given by scaling any given section fiberwise by a unit complex number $z \in U(1)$. The quotient $\mathbb{P}(\mathcal{E}) \coloneqq S(\mathcal{E})/U(1)$ is a projective space bundle over $M$ and a complex manifold. We will take $M' \coloneqq \mathbb{P}(\mathcal{E})$ to be the new phase space in this construction, along with the projection map $p_{M'} : M' \to M$.

The manifold $M'$ has a holomorphic line bundle denoted $\mathcal{O}_{\mathcal{E}}(1)$, which when restricted to any fiber $\mathbb{P}(\mathcal{E}_{x})$ over $x \in M$ is the anticanonical line bundle. We have
\[ \Gamma_{\text{hol}}(M', \mathcal{O}_{\mathcal{E}}(1)) \cong \Gamma_{\text{hol}}(M, E) \]
and more generally
\[ \Gamma_{\text{hol}}(M', \mathcal{O}_{\mathcal{E}}(1) \otimes p_{M'}^{*}\mathcal{L}) \cong \Gamma_{\text{hol}}(M, \mathcal{E} \otimes \mathcal{L})\]
where $\mathcal{L}$ is a line bundle over $M$. Thus we may take the line bundle to be $\mathcal{O}_{\mathcal{E}}(1) \otimes p_{M'}^{*}\mathcal{L}$.

Lastly, we need to describe how to construct a Hermitian metric on $\mathcal{O}_{\mathcal{E}}(1)$ from the one on $\mathcal{E}$, as well as a measure $\mu'$ on $M'$ from the measure $\mu$ on $M$ and the Hermitian metric on $\mathcal{E}$. To define a Hermitian metric on $\mathcal{O}_{E}(1)$, it will suffice to give a smoothly varying Hermitian metric restricted to each fiber, which is the line bundle $\mathcal{O}_{\mathbb{P}(\mathcal{E}_{x})}(1)$ over $p_{M'}^{-1}(x) \cong \mathbb{P}(\mathcal{E}_{x})$. This will just be the Fubini-Study metric induced by the inner product $h_{\mathcal{E}_{x}}$ on $\mathcal{E}_{x}$. Similarly, the measure $\mu'$ will be defined by the formula
\[
\int _{M'} f\,d\mu' \coloneqq \int _{M} \left(\int _{p_{M'}^{-1}(x)} f|_{p_{M'}^{-1}(x)}\,d\mu_{\text{vol}, \mathbb{P}(\mathcal{E}_{x})}\right)\,d\mu(x)
\]
where $d\mu_{\text{vol}, \mathbb{P}(\mathcal{E}_{x})}$ is the volume measure on $p_{M'}^{-1}(x)$ coming from the Fubini-Study metric, which is in turn induced by the inner product $h_{\mathcal{E}_{x}}$.

\section{A generalized de Finetti theorem from quantization}\label{sec:main-sec}
The proof of the de Finetti theorem involves first taking the Husimi Q-function of a state $\rho$, and then taking the P-quantization with respect to a subsytem. The Q-function in the first step is defined with respect to a system of coherent states on the larger Hilbert space,
while the P-quantization is defined with respect to a system of coherent states on the subsystem Hilbert space. Naturally, the proof requires analyzing how coherent states on the smaller and larger Hilbert spaces are related.
It turns out that the proof strategy of \cite{Christandl_Koenig_Mitchison_Renner_2007} and the proof of Chiribella's formula \cite{Chiribella_2011} both depend on two key properties satisfied by the coherent states:
\begin{enumerate}
\item The coherent states form a resolution of the identity. More precisely, we have
\[ \operatorname{id}_{\mathcal{H}_{\mathcal{L}}} = \int _{M} \ket{c_{\mathcal{L},x}}\bra{c_{\mathcal{L},x}}\,d\mu_{\mathcal{L}}(x) \]
for some appropriate positive measure $\mu_{\mathcal{L}}$ over $M$.
\item The coherent state on a larger Hilbert space is the tensor product of coherent states on the subsystem Hilbert spaces. For the symmetric subspace, this is more precisely stated as
\[ c_{z,k + l} = z^{\otimes (k + l)} = z^{\otimes k} \otimes z^{\otimes l} = c_{z, k} \otimes c_{z, l}. \]
An analogous property holds for other irreducible representations of $U(d)$.
\end{enumerate}

Note that in property (2), there is an implicit identification of $\mathcal{H}_{k + l}$ as a subspace of $\mathcal{H}_{k} \otimes \mathcal{H}_{l}$. In geometric quantization the two Hilbert spaces will be $\Gamma(\mathcal{L}_{1} \otimes \mathcal{L}_{2})$ and $\Gamma(\mathcal{L}_{1}) \otimes \Gamma(\mathcal{L}_{2})$. Instead of an inclusion, there is a canonical linear map
$M_{\mathcal{L}_{1},\mathcal{L}_{2}} :\Gamma(\mathcal{L}_{1}) \otimes \Gamma(\mathcal{L}_{2}) \to \Gamma(\mathcal{L}_{1} \otimes \mathcal{L}_{2})$ in the reverse direction which takes the fiberwise tensor product of two sections. The identification of one Hilbert space as a subspace of the other is equivalent to the adjoint $M_{\mathcal{L}_{1},\mathcal{L}_{2}}^{*}$ being an isometry after an appropriate normalization.

However, $M_{\mathcal{L}_{1},\mathcal{L}_{2}}^{*}$ is in general \textbf{not} an isometry. Recent work of Finski \cite[Theorem~3.16]{Finski_2022} implies that it does asymptotically becomes close to an isometry for powers of very ample line bundles. Hence with the standard definitions, the coherent states will be equal to tensor products of coherent states on subsystems only approximately.

Our approach will be slightly different from the standard construction of the Hilbert space in geometric quantization, but later will simplify the error bounds. Instead of defining the inner products on the Hilbert space using the volume form on $M$ and Hermitian metric on $\mathcal{L}$, we will assume as given a family of inner products $\langle \cdot, \cdot \rangle_{\mathcal{L}}$ on $\mathcal{H}_{\mathcal{L}}$ for every line bundle $\mathcal{L}$ such that the adjoint $M_{\mathcal{L}_{1},\mathcal{L}_{2}}^{*}$ is an exact isometry. From the perspective of classical sum-of-squares optimization, this is also more useful as it allows one to choose a computationally convenient set of inner products instead of having to compute a potentially complicated integral.

Now property (2) will hold exactly, but property (1) will only hold approximately: the measure $\mu$ which makes the coherent states into a resolution of the identity will now vary depending on the line bundle $\mathcal{L}$. Instead of needing to analyze how close a coherent state is to the tensor product of two other coherent states, we just need to analyze how close $\mu_{L_{1}}$ and $\mu_{L_{2}}$ are for two such measures, and this will turn out to be the supremum of the Radon-Nikodym derivative $d\mu_{L_{1}}/d\mu_{L_{2}}$ (roughly the ratio of the densities). In certain cases there will be exact formulas for the Radon-Nikodym derivatives which recover the earlier de Finetti theorems, while in other cases we will be able to deduce coarser asymptotics.

\subsection{Precise setup}
As before, we will fix a compact K\"{a}hler manifold $M$. However, we will now also consider a family $\mathcal{L}_{1}, \dots, \mathcal{L}_{m}$ of holomorphic line bundles over $M$, as well as all (positive) tensor product line bundles of the form
\[ \mathcal{L}_{1}^{k_{1}} \otimes \cdots \otimes \mathcal{L}_{m}^{k_{m}}. \]
We will assume that each $\mathcal{L}_{i}$ comes with some Hermitian metric, and use the tensor product Hermitian metric on the (positive) tensor product bundles. For any line bundle $\mathcal{L}$, the metric will be denoted $h_{\mathcal{L}}$.

We will also additionally assume that each quantized Hilbert space $\mathcal{H}_{\mathcal{L}}$ comes with some nondegenerate inner product $\langle \cdot, \cdot \rangle_{\mathcal{L}}$. The multiplication maps will be denoted $M_{\mathcal{L}_{1}, \mathcal{L}_{2}} : \mathcal{H}_{\mathcal{L}_{1}} \otimes \mathcal{H}_{\mathcal{L}_{2}} \to \mathcal{H}_{\mathcal{L}_{1} \otimes \mathcal{L}_{2}}$. We will assume that $M_{\mathcal{L}_{1}, \mathcal{L}_{2}}^{*}$ is an isometry with respect to the inner products on the Hilbert spaces.

We have not fixed a measure $\mu$ over $M$ yet, but will define one next. Specifically, we want $\mu$ to be chosen such that the inner product on $\mathcal{H}_{\mathcal{L}}$ induced by $\mu$ and $h_{\mathcal{L}}$ coincides with $\langle \cdot, \cdot \rangle_{\mathcal{L}}$ up to a scalar multiple. It is not immediately clear that such a measure exists, but it will follow from some known results in functional analysis.

Before constructing the measure, note that the definition of the coherent states $c_{\mathcal{L},x}$ and the function $T_{\mathcal{L}} : M \to \mathbb{R}$ from Proposition~\ref{prop:coherent-state} only use the inner product on $\mathcal{H}_{\mathcal{L}}$, and do not directly use the measure $\mu$. Thus the use of properties of the coherent states and the function $T_{\mathcal{L}}$ in our construction of $\mu$ is not circular. (Note that it would be circular to use the resolution of the identity from Proposition~\ref{prop:resolution-identity} in the construction of $\mu$, but we will not use it.)

\subsubsection{Husimi quantization rule}
Let $\mathcal{L}$ be a positive tensor product line bundle and $\{s_{i}\}$ be an orthonormal basis of $\mathcal{H}_{\mathcal{L}}$. Define the Kodaira map $\iota_{\mathcal{L}} : M \to \mathbb{P}(\mathcal{H}_{\mathcal{L}})$ by
\[ \iota_{\mathcal{L}}(x) \coloneqq c_{\mathcal{L},x}. \]
In the rest of this article, we will assume that $h_{\mathcal{L}_{i}}$ is the Hermitian metric on $\mathcal{L}$ which is the pullback under $\iota_{\mathcal{L}_{i}}$ of the Fubini-Study metric on $\mathbb{P}(\mathcal{H}_{\mathcal{L}_{i}})$ induced by $\langle \cdot, \cdot \rangle_{\mathcal{L}_{i}}$. We only make this assumption for the $\mathcal{L}_{i}$, and for their positive tensor products we take the tensor product Hermitian metric as before.

\begin{proposition}\label{prop:tl-one}
For any $\mathcal{L}_{i}$ we have $T_{\mathcal{L}_{i}}(x) = 1$, where $T_{\mathcal{L}}$ is defined as in Proposition~\ref{prop:coherent-state}.
\end{proposition}
\begin{proof}
From the definition of the Fubini-Study metric, for any orthonormal basis $s_{\alpha}$ and $x \in \mathbb{P}(\mathcal{H}_{L_{i}})$ we have
\[ \sum _{\alpha} h_{FS}(s_{\alpha}(x), s_{\alpha}(x)) = \sum _{\alpha} |\langle x, s_{\alpha} \rangle|^{2} = 1. \]
In particular, the pullback to $M$ is also the constant function 1.
\end{proof}

\begin{definition}
Let $\mathcal{L}_{1}$ and $\mathcal{L}_{2}$ be positive tensor product line bundles and $\mathcal{K} = \mathcal{L}_{1} \otimes \mathcal{L}_{2}$. The \textbf{Husimi quantization} map (or un-normalized cloning map) is defined as
\[ C_{\mathcal{L}_{1} \to \mathcal{K}}(A) \coloneqq 
M_{\mathcal{L}_{1},\mathcal{L}_{2}}(A \otimes \operatorname{id}_{\mathcal{H}_{\mathcal{L}_{2}}})M_{\mathcal{L}_{1},\mathcal{L}_{2}}^{*} \]
where $A$ is a self-adjoint operator on $\mathcal{H}_{\mathcal{L}_{1}}$.
\end{definition}

\begin{definition}
Let $\mathcal{L} = \bigotimes _{i} \mathcal{L}_{i}^{k_{i}}$ be a positive tensor product line bundle and $\mathcal{H} = \bigotimes _{i} \mathcal{H}_{\mathcal{L}_{i}}^{\otimes k_{i}}$. A function $f : M_{\mathcal{L}} \to \mathbb{R}$ has \textbf{degree $\mathcal{L}$} if there is a self-adjoint operator $A : \mathcal{H}_{\mathcal{L}} \to \mathcal{H}_{\mathcal{L}}$ such that
\[ f(x) = \sum _{\alpha, \beta} \langle s_{\alpha}, M_{\mathcal{L}}^{*}AM_{\mathcal{L}}s_{\beta}\rangle h_{\mathcal{L}}(M_{\mathcal{L}}s_{\alpha}(x), M_{\mathcal{L}}s_{\beta}(x)) \]
where $M_{\mathcal{L}} : \mathcal{H} \to \mathcal{H}_{\mathcal{L}}$ is the multplication map, and $\{s_{\alpha}\}$ is a basis of $\mathcal{H}$ which is a tensor product of orthonormal bases for each $\mathcal{H}_{\mathcal{L}_{i}}$. If additionally $A$ can be chosen to be a positive operator, then $f$ is a \textbf{degree $\mathcal{L}$ sum-of-squares}.
\end{definition}

\begin{proposition}\label{prop:poly-proj}
If $A' : \mathcal{H} \to \mathcal{H}$ is self-adjoint, then
\[ \sum _{\alpha, \beta} \langle s_{\alpha}, A's_{\beta} \rangle h_{\mathcal{L}}(M_{\mathcal{L}}s_{\alpha}(x), M_{\mathcal{L}}s_{\beta}(x)) \]
is a function of degree $\mathcal{L}$ corresponding to $M_{\mathcal{L}}A'M_{\mathcal{L}}^{*}$.
\end{proposition}
\begin{proof}
Consider the linear map sending $A'$ to the function as defined in the given expression. A given basis vector $\ket{s_{\alpha}} \bra{s_{\beta}}$ is sent to $h_{\mathcal{L}}(M_{\mathcal{L}}s_{\alpha}(x), M_{\mathcal{L}}s_{\beta}(x))$. Thus the map is the composition of $A' \mapsto M_{\mathcal{L}}A'M_{\mathcal{L}}^{*}$ and the linear map sending $\ket{s_{1}} \bra{s_{2}} \in \overline{\mathcal{H}_{\mathcal{L}}} \otimes \mathcal{H}_{\mathcal{L}}$ to $h_{\mathcal{L}}(s_{1}(x), s_{2}(x))$.
\end{proof}

\begin{proposition}\label{prop:poly-closed}
If $f$ has degree $\mathcal{L}_{1}$ and $g$ has degree $\mathcal{L}_{2}$, then $fg$ has degree $\mathcal{L}_{1} \otimes \mathcal{L}_{2}$ and $\bar{f}$ has degree $\mathcal{L}_{1}$.
\end{proposition}
\begin{proof}
For the second statement, suppose
\[ f(x) = \sum _{\alpha, \beta} \langle s_{\alpha}, M_{\mathcal{L}}AM_{\mathcal{L}}^{*}s_{\beta}\rangle_{\mathcal{L}} h_{\mathcal{L}}(s_{\alpha}(x), s_{\beta}(x)). \]
Then
\begin{align*}
\bar{f}(x) &= \sum _{\alpha, \beta} \overline{\langle s_{\alpha}, M_{\mathcal{L}}^{*}AM_{\mathcal{L}}s_{\beta}\rangle} \overline{h_{\mathcal{L}}(M_{\mathcal{L}}s_{\alpha}(x), M_{\mathcal{L}}s_{\beta}(x))} \\
&= \sum _{\alpha, \beta} \langle M_{\mathcal{L}}^{*}AM_{\mathcal{L}}s_{\beta}, s_{\alpha} \rangle h_{\mathcal{L}}(M_{\mathcal{L}}s_{\beta}(x), M_{\mathcal{L}}s_{\alpha}(x)) \\
&= \sum _{\alpha, \beta} \langle s_{\beta}, M_{\mathcal{L}}^{*}A^{*}M_{\mathcal{L}}s_{\alpha} \rangle h_{\mathcal{L}}(M_{\mathcal{L}}s_{\beta}(x), M_{\mathcal{L}}s_{\alpha}(x)) \\
&= \sum _{\alpha, \beta} \langle s_{\alpha}, M_{\mathcal{L}}^{*}A^{*}M_{\mathcal{L}}s_{\beta} \rangle h_{\mathcal{L}}(M_{\mathcal{L}}s_{\alpha}(x), M_{\mathcal{L}}s_{\beta}(x)) \\
\end{align*}
so we can take $A^{*}$ for $\bar{f}$.

For the first statement, let $A : \mathcal{H}_{\mathcal{L}_{1}} \to \mathcal{H}_{\mathcal{L}_{1}}$ and $B : \mathcal{H}_{\mathcal{L}_{2}} \to \mathcal{H}_{\mathcal{L}_{2}}$ be the corresponding operators for $f$ and $g$ respectively. Then
\begin{align*}
&\phantom{{}={}} f(x)g(x) \\
&= \sum _{\substack{\alpha_{1}, \beta_{1}\\ \alpha_{2}, \beta_{2}}} \!\langle s_{\alpha_{1}}, M_{1}^{*}AM_{1}s_{\beta_{1}} \rangle \langle s_{\alpha_{2}}, M_{2}^{*}BM_{2}s_{\beta_{2}} \rangle h_{\mathcal{L}_{1}}(M_{1}s_{\alpha_{1}}(x), M_{1}s_{\beta_{1}}(x))h_{\mathcal{L}_{2}}(M_{2}s_{\alpha_{2}}(x), M_{2}s_{\beta_{2}}(x)) \\
&= \sum _{\alpha, \beta,} \langle s_{\alpha}, (M_{1}^{*}AM_{1} \otimes M_{2}^{*}BM_{2})s_{\beta} \rangle h_{\mathcal{L}_{1} \otimes \mathcal{L}_{2}}(M_{\mathcal{L}}s_{\alpha}(x), M_{\mathcal{L}}s_{\beta}(x)) \\
&= \sum _{\alpha, \beta,} \langle s_{\alpha}, (M_{\mathcal{L}}^{*}M_{\mathcal{L}_{1},\mathcal{L}_{2}}(A \otimes B)M_{\mathcal{L}_{1},\mathcal{L}_{2}}^{*}M_{\mathcal{L}})s_{\beta} \rangle h_{\mathcal{L}_{1} \otimes \mathcal{L}_{2}}(M_{\mathcal{L}}s_{\alpha}(x), M_{\mathcal{L}}s_{\beta}(x))
\end{align*}
using Proposition~\ref{prop:poly-proj}.
\end{proof}

\begin{proposition}
The constant function 1 has degree $\mathcal{L}$ for all $\mathcal{L}$. If $f$ has degree $\mathcal{L}_{1}$, then it also has degree $\mathcal{L}_{1} \otimes \mathcal{L}_{2}$ for any positive tensor product $\mathcal{L}_{2}$.
\end{proposition}
\begin{proof}
These follow from an induction argument using Proposition~\ref{prop:tl-one} and Proposition~\ref{prop:poly-closed}.
\end{proof}

\begin{proposition}\label{prop:polys-dense}
The set of all functions with degree $\mathcal{L}$, as $\mathcal{L}$ ranges over positive tensor product line bundles, is dense in $C(M, \mathbb{C})$. The set of all functions which are degree $\mathcal{L}$ sums-of-squares is dense in $C(M, \mathbb{R}_{\geq 0})$.
\end{proposition}
\begin{proof}
Set $\mathcal{L} = \bigotimes _{i = 1} ^{m} \mathcal{L}_{i}$. By assumption, $\iota_{\mathcal{L}}$ defines a closed embedding into $\mathbb{P}(\mathcal{H}_{\mathcal{L}})$. In particular, this shows that functions of degree $\mathcal{L}$ separate points in $M$. By the previous propositions, the set of functions of degree $\mathcal{L}^{k}$ for $k \geq 0$ is an algebra closed under multiplication and conjugation. Thus by the Stone-Weierstrass theorem 
such functions are dense in $C(M, \mathbb{C})$.

Suppose that $f : M \to \mathbb{R}_{\geq 0}$ is continuous and nonnegative, and let $B = \sup _{x \in M} f(x)$. Let $\iota : M \to \mathbb{P}(\mathcal{H}_{\mathcal{L}})$ be the Kodaira map. By the Tietze extension theorem 
there is a continuous extension $g : \mathbb{P}(\mathcal{H}_{\mathcal{L}}) \to [0, B]$ such that $f = g \circ \iota$. Define
\[ l_{z}(x) \coloneqq 1 - d_{FS}(x, z) = |\langle z, x \rangle|^{2} \]
which is a sum-of-squares on $\mathbb{P}(\mathcal{H}_{\mathcal{L}})$. Let $C_{k}$ be such that $C_{k}\int _{\mathbb{P}(\mathcal{H}_{\mathcal{L}})} l_{z}(x)^{k} = 1$ and let $h_{k} : \mathbb{P}(\mathcal{H}_{\mathcal{L}}) \to \mathbb{R}$ be defined by
\[ h_{k}(z) \coloneqq C_{k}\int _{\mathbb{P}(\mathcal{H}_{\mathcal{L}})} g(x)l_{z}(x)^{k}\,d\mu_{FS}(x) \]
which is again a sum-of-squares. Let $\delta \in [0, \epsilon]$ be the modulus of uniform continuity for $g$ with respect to $d_{FS}$ (since $M$ is compact). Then
\begin{align*}
|h_{k}(z) - g(z)| &= C_{k}\left| \int _{\mathbb{P}(\mathcal{H}_{\mathcal{L}})} (g(x) - g(z))l_{z}(x)^{k}\,d\mu_{FS}(x) \right| \\
&\leq \epsilon C_{k} \int _{\{|\langle x, z \rangle| \geq \delta\}} l_{z}(x)^{k}\,d\mu_{FS}(x) + BC_{k} \int _{\{|\langle x, z \rangle| \leq \delta\}} l_{z}(x)^{k}\,d\mu_{FS}(x) \\
&\leq \epsilon + BC_{k}\delta^{k} \\
&\leq \epsilon + BC_{k}\epsilon^{k}.
\end{align*}
We know that $C_{k} = O(k^{\dim \mathcal{H}_{\mathcal{L}}})$ for $k$ large, so the right-hand side tends uniformly to 0 as $k \to \infty$. Thus $g$ is approximated uniformly by sum-of-squares polynomials. By Proposition~\ref{prop:poly-proj} the restriction of a sum-of-squares polynomial to $M$ is again a sum-of-squares polynomial, so this finishes the proof.
\end{proof}


\subsection{The generalized quantum de Finetti theorem}
In this section, we now state and prove a general version of the quantum de Finetti theorem. The statement of the theorem depends on the modified definition of coherent states, with respect to the inner products which make the multiplication maps into isometries.

\begin{theorem}\label{thm:main-thm}
Let $\rho_{12}$ be a mixed state on $\mathcal{H}_{\mathcal{L}_{1} \otimes \mathcal{L}_{2}}$ and $\rho_{1}$ be the reduced density matrix of $M_{\mathcal{L}_{1},\mathcal{L}_{2}}^{*}\rho_{12}M_{\mathcal{L}_{1},\mathcal{L}_{2}}$ over $\mathcal{H}_{1}$. Then $\rho_{1}$ is $\epsilon$-close in trace distance to a mixture over coherent states on $\mathcal{H}_{\mathcal{L}_{1}}$, where $r = d\mu_{\mathcal{L}_{2}}/d\mu_{\mathcal{L}_{1} \otimes \mathcal{L}_{2}}$ is the Radon-Nikodym derivative, $R = \sup _{x \in M} (1 - r(x))$ and $\epsilon = 2R$.
\end{theorem}

The ``tensoriztion'' property of the coherent states, which is needed to prove Theorem~\ref{thm:main-thm}, is proven in the following Proposition~\ref{prop:tensor-coh}.

\begin{proposition}\label{prop:tensor-coh}
The tensor product of two coherent states is another coherent state centered at the same point, in the sense that
\[ M_{\mathcal{L}_{1}, \mathcal{L}_{2}}e_{\mathcal{L}_{1} \otimes \mathcal{L}_{2}, x}M_{\mathcal{L}_{1}, \mathcal{L}_{2}}^{*} = e_{\mathcal{L}_{1},x} \otimes e_{\mathcal{L}_{2},x}. \]
\end{proposition}
\begin{proof}
Let $s_{1} \in \mathcal{H}_{\mathcal{L}_{1}}$ and $s_{2} \in \mathcal{H}_{\mathcal{L}_{2}}$. We have
\begin{align*}
\braket{s_{1} \otimes s_{2} | e_{\mathcal{L}_{1},x} \otimes e_{\mathcal{L}_{2},x} | s_{1} \otimes s_{2}}_{\mathcal{H}_{\mathcal{L}_{1}} \otimes \mathcal{H}_{\mathcal{L}_{2}}} &= \braket{s_{1} | e_{\mathcal{L}_{1},x} | s_{1}}_{\mathcal{L}_{1}}\braket{s_{2} | e_{\mathcal{L}_{2},x} | s_{2}}_{\mathcal{L}_{2}} \\
&= h_{\mathcal{L}_{1}}(s_{1}(x), s_{1}(x))h_{\mathcal{L}_{2}}(s_{1}(x), s_{1}(x))
\end{align*}
and
\begin{align*}
\braket{s_{1} \otimes s_{2}|M_{\mathcal{L}_{1}, \mathcal{L}_{2}}e_{\mathcal{L}_{1} \otimes \mathcal{L}_{2}, x}M_{\mathcal{L}_{1}, \mathcal{L}_{2}}^{*}|s_{1} \otimes s_{2}}
&= \braket{M_{\mathcal{L}_{1},\mathcal{L}_{2}}(s_{1} \otimes s_{2})|e_{\mathcal{L}_{1} \otimes \mathcal{L}_{2},x}|M_{\mathcal{L}_{1},\mathcal{L}_{2}}(s_{1} \otimes s_{2})} \\
&= h_{\mathcal{L}_{1} \otimes \mathcal{L}_{2}}((s_{1}\otimes s_{2})(x), (s_{1} \otimes s_{2})(x)) \\
&= h_{\mathcal{L}_{1} \otimes \mathcal{L}_{2}}(s_{1}(x) \otimes s_{2}(x), s_{1}(x) \otimes s_{2}(x)) \\
&= h_{\mathcal{L}_{1}}(s_{1}(x), s_{1}(x))h_{\mathcal{L}_{2}}(s_{1}(x), s_{1}(x)).
\end{align*}
Since set of all $s_{1} \otimes s_{2}$ spans $\mathcal{H}_{\mathcal{L}_{1}} \otimes \mathcal{H}_{\mathcal{L}_{2}}$, this is sufficient for the equality.
\end{proof}

\begin{proof}[Proof of Theorem~\ref{thm:main-thm}]
We take the classical probability distribution $P$ defining the mixture to be the Husimi $Q$-function of $\rho$. From Section~\ref{sec:qfn-coh} and Proposition~\ref{prop:tensor-coh}, this distribution is characterized by the density
\[ P(x)\,d\mu_{\mathcal{L}_{1} \otimes \mathcal{L}_{2}}(x) = \operatorname{Tr}(e_{\mathcal{L}_{1} \otimes \mathcal{L}_{2},x}\rho_{12})\, d\mu_{\mathcal{L}_{1} \otimes \mathcal{L}_{2}}(x). \]
Define
\[ \sigma \coloneqq \int_{M} e_{\mathcal{L}_{1},x}P(x)\,d\mu_{\mathcal{L}_{1} \otimes \mathcal{L}_{2}}(x) \]
to be the corresponding mixture of coherent states of $\mathcal{H}_{\mathcal{L}_{1}}$. We will show that $\sigma$ is close in trace distance to $\rho_{1}$.

To this end, we first rewrite $\rho_{1}$ as a mixture with respect to $\mu_{\mathcal{L}_{2}}$, giving
\begin{align*}
\rho_{1} &= \operatorname{Tr}_{2} ((\operatorname{id}_{\mathcal{H}_{\mathcal{L}_{1}}} \otimes \operatorname{id}_{\mathcal{H}_{\mathcal{L}_{2}}}) M_{\mathcal{L}_{1},\mathcal{L}_{2}}^{*}\rho_{12}M_{\mathcal{L}_{1},\mathcal{L}_{2}}) \\
&= \operatorname{Tr}_{2} \left(\left(\operatorname{id}_{\mathcal{H}_{\mathcal{L}_{1}}} \otimes \int _{M} e_{\mathcal{L}_{2},x}\,d\mu_{\mathcal{L}_{2}}(x)\right) \rho\right) \\
&= \int _{M} \operatorname{Tr}_{2} ((\operatorname{id}_{\mathcal{H}_{\mathcal{L}_{1}}} \otimes e_{\mathcal{L}_{2},x}) M_{\mathcal{L}_{1},\mathcal{L}_{2}}^{*}\rho_{12}M_{\mathcal{L}_{1},\mathcal{L}_{2}})\,d\mu_{\mathcal{L}_{2}}(x) \numberthis \label{eq:mainthm-povm1} \\
&= \int _{M} w(x)\,d\mu_{\mathcal{L}_{2}}(x) 
\end{align*}
where we have defined
\[ w(x) \coloneqq \operatorname{Tr}_{2} ((\operatorname{id}_{\mathcal{H}_{\mathcal{L}_{1}}} \otimes e_{\mathcal{L}_{2},x})) M_{\mathcal{L}_{1},\mathcal{L}_{2}}^{*}\rho_{12}M_{\mathcal{L}_{1},\mathcal{L}_{2}}) \in \mathcal{D}(\mathcal{H}_{\mathcal{L}_{1}}). \]

To bound the trace distance between $\rho_{1}$ and $\sigma$, we have
\begin{align*}
|\rho_{1} - \sigma| &= \left| \int _{M} w(x)\,d\mu_{\mathcal{L}_{2}}(x) - \int _{M} e_{\mathcal{L}_{1},x}P(x)\,d\mu_{\mathcal{L}_{1} \otimes \mathcal{L}_{2}}(x) \right| \\
&\leq \left| \int _{M} w(x)\,d\mu_{\mathcal{L}_{2}}(x) - \int _{M} e_{\mathcal{L}_{1},x}P(x)r(x)\,d\mu_{\mathcal{L}_{1} \otimes \mathcal{L}_{2}}(x)\right| \\
&\phantom{{}=} + \left|\int _{M} e_{\mathcal{L}_{1},x}P(x)(1 - r(x))\,d\mu_{\mathcal{L}_{1} \otimes \mathcal{L}_{2}}(x) \right|. \numberthis \label{eq:mainthm-povm2}
\end{align*}
For the second term, we have
\begin{align*}
\left|\int _{M} e_{\mathcal{L}_{1},x}P(x)(1 - r(x))\,d\mu_{\mathcal{L}_{1} \otimes \mathcal{L}_{2}}(x) \right| &\leq R \left|\int _{M} e_{\mathcal{L}_{1},x}P(x)\,d\mu_{\mathcal{L}_{1} \otimes \mathcal{L}_{2}}(x) \right| \\
&\leq R \int _{M} |e_{\mathcal{L}_{1},x}|P(x)\,d\mu_{\mathcal{L}_{1} \otimes \mathcal{L}_{2}}(x) \\
&= R/2
\end{align*}
since each coherent state has rank 1. Thus it remains to bound the first term, for which we have
\begin{align*}
&\phantom{{}={}} \left| \int _{M} w(x)\,d\mu_{\mathcal{L}_{2}}(x) - \int _{M} e_{\mathcal{L}_{1},x}P(x)r(x)\,d\mu_{\mathcal{L}_{1} \otimes \mathcal{L}_{2}}(x)\right| \\
&= \left| \int _{M} (w(x) - e_{\mathcal{L}_{1},x}P(x))\,d\mu_{\mathcal{L}_{2}}(x) \right| \\
&= \left| \int _{M} (w(x) - e_{\mathcal{L}_{1},x}\operatorname{Tr}(e_{\mathcal{L}_{1} \otimes \mathcal{L}_{2},x}\rho_{12})\,d\mu_{\mathcal{L}_{2}}(x) \right| \\
&= \left| \int _{M} (w(x) - e_{\mathcal{L}_{1},x}\operatorname{Tr}((e_{\mathcal{L}_{1},x} \otimes e_{\mathcal{L}_{2},x})M^{*}_{\mathcal{L}_{1} \otimes \mathcal{L}_{2}}\rho_{12}M_{\mathcal{L}_{1} \otimes \mathcal{L}_{2}})\,d\mu_{\mathcal{L}_{2}}(x) \right| \\
&= \left| \int _{M} (w(x) - e_{\mathcal{L}_{1},x}\operatorname{Tr}(e_{\mathcal{L}_{1},x} w(x)))\,d\mu_{\mathcal{L}_{2}}(x) \right| \\
&= \left| \int _{M} (w(x) - e_{\mathcal{L}_{1},x}w(x)e_{\mathcal{L}_{1},x})\,d\mu_{\mathcal{L}_{2}}(x) \right|.
\end{align*}

Using the identity
\[ A - BAB = (A - BA) + (A - AB) - (1 - B)A(1 - B) \] 
for square matrices $A$ and $B$ (also used in the CKMR proof), we have
\[ \left| \int _{M} (w(x) - e_{\mathcal{L}_{1},x}w(x)e_{\mathcal{L}_{1},x})\,d\mu_{\mathcal{L}_{2}}(x) \right| \leq |\alpha| + |\beta| + |\gamma| \]
where
\begin{align*}
\alpha &\coloneqq \int _{M} (w(x) - e_{\mathcal{L}_{1},x}w(x))\,d\mu_{\mathcal{L}_{2}}(x) \\
\beta &\coloneqq \int _{M} (w(x) - w(x)e_{\mathcal{L}_{1},x})\,d\mu_{\mathcal{L}_{2}}(x) \\
\gamma &\coloneqq \int _{M} (\operatorname{id} - e_{\mathcal{L}_{1},x})w(x)(\operatorname{id} - e_{\mathcal{L}_{1},x})\,d\mu_{\mathcal{L}_{2}}(x).
\end{align*}
For convenience, define $\rho' \coloneqq M_{\mathcal{L}_{1},\mathcal{L}_{2}}^{*}\rho_{12}M_{\mathcal{L}_{1},\mathcal{L}_{2}}$. We then have
\begin{align*}
&\phantom{{}={}} |\alpha| \\
&= \left| \rho_{1} - \int _{M} e_{\mathcal{L}_{1},x}\operatorname{Tr}_{2} ((\operatorname{id}_{\mathcal{H}_{\mathcal{L}_{1}}} \otimes e_{\mathcal{L}_{2},x}) \rho') \, d\mu_{\mathcal{L}_{2}}(x) \right| \\
&= \left| \operatorname{Tr}_{2}(\rho') - \int _{M} \operatorname{Tr}_{2} ((e_{\mathcal{L}_{1},x} \otimes e_{\mathcal{L}_{2},x})\rho') \, d\mu_{\mathcal{L}_{2}}(x) \right| \\
&= \left| \int _{M} \operatorname{Tr}_{2}((e_{\mathcal{L}_{1},x} \otimes e_{\mathcal{L}_{2},x})\rho')\,d\mu_{\mathcal{L}_{1} \otimes \mathcal{L}_{2}}(x) - \int _{M} \operatorname{Tr}_{2} ((e_{\mathcal{L}_{1},x} \otimes e_{\mathcal{L}_{2},x}) \rho')) \, d\mu_{\mathcal{L}_{2}}(x) \right| \numberthis \label{eq:mainthm-povm3} \\
&= \left| \int _{M} \operatorname{Tr}_{2}((e_{\mathcal{L}_{1},x} \otimes e_{\mathcal{L}_{2},x})\rho') (1 - r(x))\,d\mu_{\mathcal{L}_{1} \otimes \mathcal{L}_{2}}(x) \right| \\
&\leq R \int _{M} \operatorname{Tr}_{2}((e_{\mathcal{L}_{1},x} \otimes e_{\mathcal{L}_{2},x})\rho') \,d\mu_{\mathcal{L}_{1} \otimes \mathcal{L}_{2}}(x) \\
&= R/2 |\rho_{1}| \\
&= R/2.
\end{align*}
Similarly, we can show
\[ w(x) = \operatorname{Tr}_{2} ( M_{\mathcal{L}_{1},\mathcal{L}_{2}}^{*}\rho_{12}M_{\mathcal{L}_{1},\mathcal{L}_{2}}(\operatorname{id}_{\mathcal{H}_{\mathcal{L}_{1}}} \otimes e_{\mathcal{L}_{2},x})) \]
and then by a similar argument show that $|\beta| \leq R/2$. For any $\xi \succeq 0$ and projector $P$, we have $P^{*}\xi P \succeq 0$, so
\[ |P^{*} \xi P| = \frac{1}{2}\operatorname{Tr}(P \xi P) = \frac{1}{2} \operatorname{Tr} (P \xi). \]
Thus
\begin{align*}
|\gamma| &\leq \int _{M} |(\operatorname{id} - e_{\mathcal{L}_{1},x})w(x)(\operatorname{id} - e_{\mathcal{L}_{1},x})|\,d\mu_{\mathcal{L}_{2}}(x) \\
&= \frac{1}{2} \operatorname{Tr} \left( \int _{M} (\operatorname{id} - e_{\mathcal{L}_{1},x})w(x)(\operatorname{id} - e_{\mathcal{L}_{1},x})\,d\mu_{\mathcal{L}_{2}}(x) \right) \\
&= \frac{1}{2} \operatorname{Tr}(\alpha) \\
&\leq |\alpha| \\
&\leq R/2. \numberthis \label{eq:mainthm-povm4}
\end{align*}
Combining everything, this gives $|\rho_{1} - \sigma| \leq 2R$.
\end{proof}

\subsection{Explicit formulas for homogeneous spaces}
Our statement of the generalized quantum de Finetti theorem (Theorem~\ref{thm:main-thm}) requires bounds the error in terms of two measures which make the two families of coherent states into POVMs. Later we will show that such measures always exist approximately. However, the resulting measures are not necessarily explicit and a simple formula seems unlikely to exist in general. In the case when the $M$ and $\mathcal{L}_{1}, \dots, \mathcal{L}_{m}$ are suitably symmetric, we will show that the measure can be chosen to be a fixed scalar multiple of the invariant measure on $M$, which is unique. The Radon-Nikodym of two such measures becomes the ratio of the two scalars, which is easier to calculate. While the symmetry condition may seem restrictive, many of the earlier quantum de Finetti theorems can be deduced as special cases when this symmetry condition holds. The details for several such de Finetti theorems are covered in Section~\ref{sec:applications}.

The precise symmetry condition we use is
\begin{definition}
The pair $(M, \{\mathcal{L}_{1}, \dots, \mathcal{L}_{m}\})$ is \textbf{homogeneous} if there is a compact group $G$ acting transitively on $M$ by K\"{a}hler automorphisms and acting on each $\mathcal{L}_{i}$ by bundle automorphisms, such that the hermitian metric on $\mathcal{L}_{i}$ is invariant under $G$ and the inner product on $\mathcal{H}_{\mathcal{L}}$ is also invariant under $G$.
\end{definition}

\begin{proposition}\label{prop:coh-equivariant}
Suppose that $M$ and $\mathcal{L}$ are homogeneous with respect to $G$. Then $c_{\mathcal{L},g\cdot z} = g \cdot c_{\mathcal{L}, z}$ for all $z \in M$.
\end{proposition}
\begin{proof}
For any pure state $\psi$, we have
\begin{align*}
\operatorname{Tr}(c_{\mathcal{L},g\cdot z} \psi) &= |\psi(g \cdot z)|^{2} \\
&= |(g^{-1} \cdot \psi)(z)|^{2} \\
&= \operatorname{Tr}(c_{\mathcal{L}, z} (g^{-1} \psi g)) \\
&= \operatorname{Tr}((g \cdot c_{\mathcal{L}, z}) \psi).
\end{align*}
\end{proof}

\begin{proposition}
Suppose that $M$ and $\mathcal{L}$ are homogeneous with respect to $G$. Then the Hilbert space $\mathcal{H}_{\mathcal{L}}$ is an irreducible representation of $G$.
\end{proposition}
\begin{proof}
Suppose $\mathcal{H}' \subseteq \mathcal{H}_{\mathcal{L}}$ is a $G$-invariant subspace. Let $P$ be the orthogonal projection onto $\mathcal{H}'$. Then
\begin{align*}
\operatorname{Tr}(P c_{\mathcal{L},g \cdot z}) &= \operatorname{Tr}(P (g \cdot c_{\mathcal{L},z})) \\
&= \operatorname{Tr}((g^{-1}Pg) c_{\mathcal{L},z}) \\
&= \operatorname{Tr}(P c_{\mathcal{L},z})
\end{align*}
for all $g \in G$ and $z \in M$. Since $G$ acts transitively on $M$, this implies that the function $z \mapsto \operatorname{Tr}(P c_{\mathcal{L},z})$ is constant. By Proposition~\ref{prop:coh-span} the coherent states $c_{\mathcal{L},z}$ span $\mathcal{H}_{\mathcal{L}}$, so $P$ is a constant multiple of the identity. Thus $\mathcal{H}'$ is either $0$ or $\mathcal{H}_{\mathcal{L}}$.
\end{proof}

\begin{proposition}
Suppose that $M$ and $\mathcal{L}$ are homogeneous with respect to $G$. Let $\mu_{\text{vol}}$ be the volume measure on $M$ normalized to have mass 1, which is the unique $G$-invariant probability measure on $M$. Then taking $\mu_{\mathcal{L}} = (\dim \mathcal{H}_{\mathcal{L}})\mu_{\text{vol}}$ makes the coherent states into a normalized POVM.
\end{proposition}
\begin{proof}
We have
\begin{align*}
g \cdot \left(\int _{M} c_{\mathcal{L},z}\,d\mu_{\text{vol}}(z)\right) &= \int _{M} (g \cdot c_{\mathcal{L},z})\,d\mu_{\text{vol}}(z) \\
&= \int _{M} c_{\mathcal{L},g \cdot z}\,d\mu_{\text{vol}}(z) \\
&= \int _{M} c_{\mathcal{L},z}\,d\mu_{\text{vol}}(z)
\end{align*}
using Proposition~\ref{prop:coh-equivariant} and the $G$-invariance of $\mu_{\text{vol}}$. Since the inner product on $\mathcal{H}_{\mathcal{L}}$ is $G$-invariant, by Schur's lemma the integral must be a scalar multiple of $\operatorname{id}_{\mathcal{H}_{\mathcal{L}}}$. Then
\[ \operatorname{Tr} \left(\int _{M} c_{\mathcal{L},z}\,d\mu_{\text{vol}}(z) \right) = \int _{M} \operatorname{Tr}(c_{\mathcal{L},z})\,d\mu_{\text{vol}}(z) = 1 \]
so the integral is $(\dim \mathcal{H}_{\mathcal{L}})^{-1} \operatorname{id}_{\mathcal{H}_{\mathcal{L}}}$. Rescaling $\mu_{\text{vol}}$ by $\dim \mathcal{H}_{\mathcal{L}}$ then makes the coherent states into a normalized POVM.
\end{proof}

\begin{corollary}\label{cor:homog-err}
If $M$ and $\mathcal{L}_{1}, \dots, \mathcal{L}_{m}$ are homogeneous, then the error in the generalized de Finetti theorem (Theorem~\ref{thm:main-thm}) is bounded by $1 - (\dim \mathcal{H}_{\mathcal{L}_{1}})/(\dim \mathcal{H}_{\mathcal{L}_{2}})$.
\end{corollary}



\section{Applications}\label{sec:applications}
\subsection{The original quantum de Finetti theorem}
We begin by showing how the generalized quantum de Finetti theorem, Theorem~\ref{thm:main-thm}, has the usual quantum de Finetti theorem as a special case. We will work this case out in slightly more detail than the other applications. To deduce the original quantum de Finetti theorem we will take $M = \mathbb{C}P^{d - 1}$, which is the manifold of pure states in a $d$-dimensional Hilbert space. We will take $\mathcal{L}$ to be the dual of the canonical line bundle over $\mathbb{C}P^{d - 1}$, where the canonical line bundle has the total space
\[ \{ (x, v) : x \in \mathbb{C}P^{d - 1}, v \in \operatorname{span} _{\mathbb{C}} \{z\} \}. \]
Equivalently, the fiber over a point $z \in \mathbb{C}P^{d - 1}$ the 1-dimensional vector space $\operatorname{span} _{\mathbb{C}} \{z\}$. We will take the hermitian metric on $\mathcal{L}$ to be the usual Fubini-Study metric with respect to the standard hermitian inner product on $\mathbb{C}^{d}$.

\subsubsection{Explicit charts for complex projective space}
We first explicitly describe the K\"{a}hler structure on $\mathbb{C}P^{d - 1}$. Using homogeneous coordinates, let $[ z_{1} : \cdots : z_{d} ]$ denote the point in $\mathbb{C}P^{d - 1}$ corresponding to the subspace spanned by $(z_{1}, \dots, z_{d})$. An explicit family of charts are defined on the open subsets
\[ U_{i} \coloneqq \{[z_{1} : \cdots : z_{d} ] \,\mid\, z_{i} \neq 0\} \]
with maps $\phi_{i} : U_{i} \to \mathbb{C}^{d - 1}$ defined by
\[ \phi_{i}([z_{1} : \cdots : z_{d}]) \coloneqq \frac{1}{z_{i}} (z_{1}, \dots, z_{i - 1}, z_{i + 1}, \dots, z_{d}). \]
For notational convenience, we will sometimes take the range of $\phi_{i}$ to be $\{ z \in \mathbb{C}^{d} : z_{i} = 1\}$, inserting a 1 at the $i$-th position in the equation above. With this convention, we can calculate that the transition maps $\phi_{i \to j}$ between the charts are
\[ \phi_{i \to j}(w) = \frac{w_{i}}{w_{j}}w \]
which are holomorphic on their domains. To define local trivializations of the canonical line bundle, we will use the same domains $U_{i}$. The local trivialization $\phi'_{i} : \pi_{\mathcal{L}}^{-1}(U_{i}) \to \mathbb{C}^{d - 1} \times \mathbb{C}$ is given by
\[ \phi'_{i}(w, v) = (\phi_{i}(w), v_{i}) \]
and the transition maps are given by
\[ \phi'_{i \to j}(w, v) = \left(\phi_{i \to j}(w), \frac{w_{i}}{w_{j}}v\right). \]
For the line bundle $\mathcal{L}$, we then obtain local trivializations
\[ \phi''_{i}(w, v) = (\phi_{i}(w), 1/v_{i}) \]
and transition maps
\[ \phi''_{i \to j}(w, v) = \left(\phi_{i \to j}(w), \frac{w_{j}}{w_{i}}v\right). \]
Lastly, the Fubini-Study hermitian metric on $\mathcal{L}$ is a smooth global section of $\bar{\mathcal{L}} \otimes \mathcal{L}$ given in the chart $U_{i}$ by
\[ h_{\mathcal{L},z}(v, v') = \frac{\bar{v}v'}{1 + |z|^{2}}.\]

\subsubsection{The disk bundle point of view}
There is an alternative description of the Hilbert spaces as subspaces of a single larger Hilbert space. For a slightly more detailed reference, see \cite[Section~2]{Zelditch_1998}. Consider $\mathcal{L}^{*}$, which is the canonical bundle over $\mathbb{C}P^{d}$, with the dual hermitian metric. The unit circle bundle of $\mathcal{L}^{*}$ is
\[ \{ (z, a) \in \mathcal{L}^{*} : |a|_{h_{\mathcal{L}^{*}}} = 1 \} \]
which turns out to be the same as the unit sphere
\[ \{ z \in \mathbb{C}^{d} : |z| = 1\} = S^{2d - 1}. \]
The unit sphere is a real-analytic submanifold of $\mathbb{C}^{d}$, and thus it has the structure of a Cauchy-Riemann manifold (CR manifold). In particular, this gives a notion of holomorphic function over $S^{2d - 1}$, which are restrictions of holomorphic functions defined on the open unit disk bundle whose boundary is the circle bundle. We can then consider the Hilbert space
\[ L^{2}_{\text{hol}}(S^{2d - 1}) \coloneqq \left\{ \text{$f : S^{2d - 1} \to \mathbb{C}$ holomorphic} \middle| \int _{S^{2d - 1}} |f|^{2}\,d\mu < \infty \right\} \]
where $\mu$ is now the $U(d)$-invariant measure on $S^{2d - 1}$. In particular, $\mu$ is invariant under the $U(1)$ action which acts fiberwise (thinking of $S^{2d - 1}$ as a circle bundle over $\mathbb{C}P^{d - 1}$), or which acts by scalar multiplication (thinking of $S^{2d - 1}$ as a subset of $\mathbb{C}^{d}$).
Then there is a $U(1)$-action on $L^{2}_{\text{hol}}(S^{2d - 1})$, and one can consider the decomposition into irreducible representations of $U(1)$. It is known that the space of holomorphic sections of Fourier weight $k$ is exactly $\mathcal{H}_{\mathcal{L}^{k}}$ \cite{Zelditch_1998}.

In the case of $\mathbb{C}P^{d - 1}$, we use this to explicitly determine $\mathcal{H}_{\mathcal{L}^{k}}$. We see that $L^{2}_{\text{hol}}(S^{2d - 1})$ is spanned by all monomials $z^{\alpha}$ for $\alpha \in \mathbb{Z}_{\geq 0}^{d}$. If $a \in U(1)$, then
\[ a \cdot z^{\alpha} = (az)^{\alpha} = a^{|\alpha|}z^{\alpha} \]
so $z^{\alpha}$ has Fourier weight equal to the total degree $|\alpha|$. Thus $\mathcal{H}_{L^{k}}$ is spanned by monomials of total degree $k$, which can be identified with the symmetric subspace $\operatorname{Sym}^{k}(\mathbb{C}^{d})$.

\subsubsection{Deducing the quantum de Finetti theorem}
We will treat the pair $(\mathbb{C}P^{d - 1}, \mathcal{L})$ as being homogeneous under the group action of the unitary group $U(d)$, and then take $\mu$ to be the unitarily invariant measure.

\begin{proposition}
The unitary group $U(d)$ acts transitively on the pair $(\mathbb{C}P^{d - 1}, \mathcal{L})$.
\end{proposition}
\begin{proof}
We define the $U(d)$ action by $g \cdot vv^{*} = (gv)(gv)^{*}$, where $vv^{*}$ is the projector onto a 1-dimensional subspace of $\mathbb{C}^{d}$ (which represents a point in $\mathbb{C}P^{d - 1}$). This extends to an action on the total space of the canonical bundle by
\[ g \cdot (x, v) \coloneqq (g \cdot x, gv) \]
where $gv$ denotes applying a unitary to a vector.

To show that the action is transitive, let $z_{1}, z_{2} \in \mathbb{C}^{d}$ be any two unit vectors. There exist two orthonormal bases of $\mathbb{C}^{d}$ enlarging the sets $\{z_{1}\}$ and $\{z_{2}\}$ respectively. The unitary which takes one basis to the other will then take $z_{1}$ to $z_{2}$, and thus take the corresponding points of $\mathbb{C}^{d - 1}$ to each other.

To show that the Fubini-Study metric is invariant under $G$, it suffices to consider the subgroups $G_{i} \cong U(d - 1)$ which preserve the $i$-th coordinate, since these subgroups generate all of $U(d)$. For $g \in G_{i}$, we have $\phi_{i}(g \cdot w) = g \cdot \phi_{i}(w)$. Thus
\[ h_{\mathcal{L}, g \cdot z}(v, v') = \frac{\bar{v}v'}{1 + |g \cdot z|^{2}} = \frac{\bar{v}v'}{1 + |z|^{2}} = h_{\mathcal{L}, z}(v, v') \]
where the middle equality uses the fact that $g$ preserves the $i$-th coordinate.
\end{proof}

The original quantum de Finetti theorem is stated for states on the symmetric subspace, so to recover it we will need to show that this agrees with the Hilbert space constructed in geometric quantization.

\begin{proposition}\label{prop:sym-subspace}
The Hilbert space $\mathcal{H}_{\mathcal{L}^{k}}$ is isomorphic to the symmetric subspace $\operatorname{Sym}^{k} \mathbb{C}^{d}$.
\end{proposition}
\begin{proof}
Let $s \in \mathcal{H}_{\mathcal{L}^{k}}$ be global section. In the chart $U_{i}$, the restriction of $s$ is some holomorphic function $s_{i} : \mathbb{C}^{d - 1} \to \mathbb{C}$. Using the transition map between $U_{i}$ and $U_{j}$, we have
\begin{align*}
s_{j}(z_{1}, \dots, z_{j - 1}, z_{j + 1}, \dots, z_{d}) = \left(\frac{z_{i}}{z_{j}}\right)^{k} s_{i}\left( \frac{z_{j}}{z_{i}} (z_{1}, \dots, z_{i - 1}, z_{i + 1}, \dots, z_{d}) \right)
\end{align*}
on the appropriate domain. In order for each $s_{i}$ to be holomorphic on all of $\mathbb{C}^{d - 1}$, we must have that
\[ z_{i}^{k}s_{i}(z_{1}/z_{i}, \dots, z_{i - 1}/z_{i}, z_{i + 1}/z_{i}, \dots, z_{d}/z_{i}) \]
is a homogeneous polynomial of degree $k$ which does not depend on $i$. Thus the isomorphism can be taken to be the linear map which sends $s \in \mathcal{H}_{\mathcal{L}^{k}}$ to the corresponding homogeneous polynomial of degree $k$, which can be identified with an element of the symmetric subspace.
\end{proof}

\begin{proposition}
The coherent state at a point $z \in \mathbb{C}P^{d - 1}$ is $z^{\otimes k} \in \mathcal{H}_{\mathcal{L}^{k}} \subseteq (\mathbb{C}^{d})^{\otimes k}$, where we identify the coherent state with its images under the multiplication maps.
\end{proposition}
\begin{proof}
By applying Proposition~\ref{prop:tensor-coh} and taking tensor products, it suffices to prove the statement for $k = 1$. Let $w \in \mathbb{C}^{d}$ be a unit vector $i$ be some index such that $w_{i} \neq 0$. Let $s \in \mathcal{H}_{\mathcal{L}}$ be some section. By Proposition~\ref{prop:sym-subspace}, there is some $a$ such that $s(z) = a \cdot z$ in the chart $U_{i}$, using the convention that $z_{i} = 1$. Calculating in the chart $U_{i}$, we have
\[ |s(w)|^{2} = \frac{1}{|w/w_{i}|^{2}} |a \cdot (w/w_{i})|^{2} = |a \cdot w|^{2} \]
finishing the proof.
\end{proof}

\begin{proposition}\label{prop:origdf}
If $\rho$ is a mixed state on $\operatorname{Sym}^{n}(\mathbb{C}^{d})$, then $\operatorname{Tr}_{n - k} \rho$ is $\epsilon$-close in trace distance to a mixture over coherent states, where $\epsilon \leq 2dk/n$.
\end{proposition}
\begin{proof}
We apply Theorem~\ref{thm:main-thm} and Corollary~\ref{cor:homog-err} to $\mathbb{C}P^{d - 1}$ and $\mathcal{L}$.  We get the statement for some $\epsilon$ satisfying
\begin{align*}
\epsilon &\leq 2\left(1 - \frac{\dim \operatorname{Sym}^{n - k}(\mathbb{C}^{d})}{\dim \operatorname{Sym}^{n}(\mathbb{C}^{d})} \right) \\
&= 2\left(1 - \frac{\binom{n - k + d - 1}{d - 1}}{\binom{n + d - 1}{d - 1}} \right).
\end{align*}
Then
\begin{align*}
\frac{\binom{n - k + d - 1}{d - 1}}{\binom{n + d - 1}{d - 1}} &= \frac{(n - k + d - 1)!/((n - k)!(d - 1)!)}{(n + d - 1)!/(n!(d - 1)!)} \\
&= \frac{(n - k + d - 1)!/(n - k)!}{(n + d - 1)!/n!} \\
&= \prod _{i = 1} ^{d - 1} \frac{n - k + i}{n + i} \\
&\geq \left( \frac{n - k + 1}{n + 1} \right)^{d - 1} \\
&= \left(1 - \frac{k}{n + 1}\right)^{d - 1} \\
&\geq 1 - \frac{(d - 1)k}{n + 1} \\
&\geq 1 - \frac{dk}{n}
\end{align*}
using the binomial theorem in second-to-last line, so $\epsilon \leq 2dk/n$.
\end{proof}

\subsection{A multi-symmetric quantum de Finetti theorem}\label{sec:multisym}
One simple way to construct new K\"{a}hler manifolds from existing ones is by taking products. Standard results \cite{Wells_2008} show that the product of two compact K\"{a}hler manifolds is again a compact K\"{a}hler manifold in a canonical way. In particular, if $M_{1}$ and $M_{2}$ are compact K\"{a}hler manifolds each with lines bundles $\mathcal{L}_{1}$ and $\mathcal{L}_{2}$, then we have two maps
\begin{align*}
p_{1} : M_{1} \times M_{2} \to M_{1} &&  p_{2} : M_{1} \times M_{2} \to M_{2}
\end{align*}
which project onto the first and second coordinate respectively, and we also have the corresponding pullback bundles $p_{1}^{*}\mathcal{L}_{1}$ and $p_{2}^{*}\mathcal{L}_{2}$. The \textbf{external tensor product} of $\mathcal{L}_{1}$ and $\mathcal{L}_{2}$ is defined to be
\[ \mathcal{L}_{1} \boxtimes \mathcal{L}_{2} \coloneqq p_{1}^{*}\mathcal{L}_{1} \otimes p_{2}^{*}\mathcal{L}_{2}. \]

\begin{remark}
The external tensor product should not be confused with the standard tensor product of two line bundles over the same manifold. The relation between them is as follows: if $\mathcal{L}_{1}$ and $\mathcal{L}_{2}$ is line bundles over the same space $M$, then $\mathcal{L}_{1} \boxtimes \mathcal{L}_{2}$ is a line bundle over $M \times M$. The restriction of $\mathcal{L}_{1} \boxtimes \mathcal{L}_{2}$ to the diagonal subset $\{(m, m) : m \in M\} \cong M$ is isomorphic to $\mathcal{L}_{1} \otimes \mathcal{L}_{2}$ over $M$.
\end{remark}

We will not prove the following result, but the proof can be found in standard complex geometry texts \cite{Wells_2008}.

\begin{proposition} \label{prop:tensor-hilb}
Let $\mathcal{L}_{i}$ be a holomorphic line bundle over $M_{i}$ for $i \in [n]$, and assume that all $M_{i}$ are compact. Denote $\mathcal{L} = \boxtimes _{i = 1} ^{n} \mathcal{L}_{i}$. Then
\[ \mathcal{H}_{\mathcal{L}} \cong \bigotimes _{i = 1} ^{n} \mathcal{H}_{\mathcal{L}_{i}}. \]
\end{proposition}

\begin{proposition}
For $i \in [n]$ let $z_{i} \in M_{i}$ and $c_{\mathcal{L}_{i},z_{i}}$ be the corresponding coherent state for some line bundles $\mathcal{L}_{i}$ over $M_{i}$. Denote $M = \prod _{i = 1} ^{n} M_{i}$, $\mathcal{L} = \boxtimes _{i = 1} ^{n} \mathcal{L}_{i}$, and $z = (z_{1}, \dots, z_{n})$. Then $c_{\mathcal{L},z} = \bigotimes_{i = 1} ^{n} c_{\mathcal{L}_{i},z_{i}}$.
\end{proposition}
\begin{proof}
Consider sections $s_{i} \in \mathcal{H}_{\mathcal{L}_{i}}$ for $i \in [n]$. Denote $s = \bigotimes _{i = 1} ^{n} s_{i}$. Then
\begin{align*}
\left|\left\langle s, \bigotimes _{i = 1} ^{n} c_{\mathcal{L}_{i},z_{i}} \right\rangle_{\mathcal{L}}\right|^{2} &= \prod _{i = 1} ^{n} |\langle s_{i}, c_{\mathcal{L}_{i},z_{i}} \rangle_{\mathcal{L}_{i}}|^{2} \\
&= \prod _{i = 1} ^{n} h_{\mathcal{L}_{i}}(s_{i}(z_{i}), s_{i}(z_{i})) \\
&= h_{\mathcal{L}}(s, z).
\end{align*}
The coherent states are uniquely characterized by the linear functional they define, and by Proposition~\ref{prop:tensor-hilb} the tensor product sections span $\mathcal{H}_{\mathcal{L}}$. Thus the equality follows.
\end{proof}

Some standard facts on group actions give the following proposition for homogeneous spaces.

\begin{proposition}
Suppose for $i \in [n]$ that $M_{i}$ and $\mathcal{L}_{i}$ are homogeneous for a $G_{i}$-action on $M_{i}$. Then $\prod _{i = 1} ^{n} M_{i}$ and $\boxtimes_{i = 1} ^{n} \mathcal{L}_{i}$ are homogeneous for the action of $\prod _{i = 1} ^{n} G_{i}$.
\end{proposition}

By combining this with the proof of the original quantum de Finetti theorem, we get the following ``multi-symmetric'' quantum de Finetti theorem. A two-sided version of the quantum de Finetti theorem was studied in \cite{Jakab_Solymos_Zimboras_2022}. A similar result to ours was given in \cite[Theorem~14]{Johnston_Lovitz_Vijayaraghavan_2023} based on the general theorem in \cite{Koenig_Mitchison_2009} with a proof closer to the original proof of the quantum de Finetti theorem. For qubits the coherent states can also be described in spherical coordinates on the Bloch sphere, which are known as Bloch coherent states. A version of this de Finetti theorem written in terms of Bloch coherent states was also shown in \cite{Lieb_1973}. For simplicity of notation we prove our result only when certain indices and parameters are equal, but one can easily generalize it to recover the results of \cite{Johnston_Lovitz_Vijayaraghavan_2023} or \cite{Lieb_1973}.

\begin{proposition}\label{prop:multisym}
Suppose $\rho$ is a quantum state on $m$ registers, where the $j$-th register contains $n$ qudits of dimension $d$, which is simultaneously invariant under permutations of the $n$ qudits within each register. Then the reduced state of $\rho$ on the first $k$ qudits within each register
is $\epsilon$-close to a mixture over coherent states, where
$\epsilon \leq 2\frac{mdk}{n}$.
\end{proposition}
\begin{proof}
We apply Theorem~\ref{thm:main-thm} and Corollary~\ref{cor:homog-err} to $(\mathbb{C}P^{d - 1})^{m}$ and the line bundles $\mathcal{L}_{2} \coloneqq \bigotimes _{i = 1} ^{n} (p_{i}^{*}\mathcal{L}_{i})^{n}$ and $\mathcal{L}_{1} \coloneqq \bigotimes _{i = 1} ^{n} (p_{i}^{*}\mathcal{L}_{i})^{n - k}$.  We get the statement for some $\epsilon$ satisfying
\begin{align*}
\epsilon &\leq 2\left(1 - \frac{\dim \mathcal{H}_{\mathcal{L}_{1}}}{\dim \mathcal{H}_{\mathcal{L}_{2}}}\right) \\
&= 2\left(1 - \frac{\dim (\operatorname{Sym}^{n - k}(\mathbb{C}^{d}))^{\otimes m}}{\dim (\operatorname{Sym}^{n}(\mathbb{C}^{d}))^{\otimes m}} \right) \\
&= 2\left(1 - \left(\frac{\dim \operatorname{Sym}^{n - k}(\mathbb{C}^{d})}{\dim \operatorname{Sym}^{n}(\mathbb{C}^{d})}\right)^{m} \right).
\end{align*}
From the proof of \ref{prop:origdf}, we then get
\begin{align*}
\left(\frac{\dim \operatorname{Sym}^{n - k}(\mathbb{C}^{d})}{\dim \operatorname{Sym}^{n}(\mathbb{C}^{d})}\right)^{m} &\geq \left(1 - \frac{dk}{n}\right)^{m} \\
&\geq 1 - \frac{mdk}{n}
\end{align*}
so $\epsilon \leq 2\frac{mdk}{n}$.
\end{proof}

\subsubsection{Integrality gaps for ncSoS relaxations of \textsc{Quantum Max-Cut} and \textsc{Quantum Max-\texorpdfstring{$d$}{d}-Cut}}
As an application of the multi-symmetric de Finetti theorem, we will show that the quantum Heisenberg model on a sufficiently symmetric graph approximates the classical Heisenberg model. From this, one can construct integrality gaps for SDP relaxations of \textsc{Quantum Max-Cut} and \textsc{Quantum Max-$d$-Cut} from integrality gaps for relaxations of the limiting classical problems. The technique will in fact also give some hardness of approximation results.

Previously, it was known that the integrality gap for \textsc{Quantum Max-Cut} is a particular constant related to the integrality gap for the rank-3 Grothendieck problem \cite{Hwang_Neeman_Parekh_Thompson_Wright_2023}, but only assuming an isoperimetric-type inequality which is still conjectural. That inequality implies that a specific family of instances gives the integrality gaps for the level-1 ncSoS relaxation of \textsc{Quantum Max-Cut}. The result shown here does not make such an assumption, but also does not determine the precise numerical value of the integrality gap.

\begin{proposition}\label{prop:maxdcut-approx}
Let $G$ be a weighted graph and $G' = K_{k} \times G$ be the cartesian product of $K_{k}$ with $G$. Then the maximum energy of the \textsc{Quantum Max-$d$-Cut} Hamiltonian on $G'$ is contained in the interval $[C, C + O(1/k)]$, where $C$ is the optimal solution of the rank-$(2d - 1)$ Grothendieck problem on $G$.
\end{proposition}

\begin{remark}
One can visualize $G'$ as being $G$ where each vertex $v$ has been replaced by a ``cloud'' $\{v\} \times [k]$ of $k$ vertices, and an edge connects $(v, i)$ to $(w, j)$ with weight $w/k^{2}$ whenever $(v, w) \in E(G)$ has weight $w$.
\end{remark}

\begin{proof}
Let $H_{G'}$ be the \textsc{Quantum Max-$d$-Cut} Hamiltonian for $G'$. Since $H_{G'}$ is invariant under the $(S_{k})^{n}$ action where the $i$-th permutation permutes the $i$-th cloud, it is supported on the subspace $\mathcal{H}' \coloneqq (\operatorname{Sym}^{k}(\mathbb{C}^{d}))^{\otimes n} \subseteq (\mathbb{C}^{d})^{\otimes kn}$. Moreover, let $H_{v, w,1} \coloneqq 1 - \sum _{\alpha = 1} ^{d^{2}} (X_{\alpha})_{v,1}(X_{\alpha})_{w,1}$ be the term corresponding to an edge $((v, 1), (w, 1))$. Then the sum of terms corresponding to $\{((v, i), (w, j)) : i, j \in [k]\}$ is
\begin{align*}
\sum _{\sigma, \sigma', \tau, \tau' \in S_{k}} \sigma'_{v}\tau'_{w}H_{v,w,1}\sigma_{v}^{*}\tau_{w}^{*} &= \sum _{g \in (S_{k})^{n}} gH_{v,w,1}g^{*} \\
&= P_{\mathcal{H}'}H_{v,w,1}P_{\mathcal{H}'}.
\end{align*}
By summing over edges in $G$, if we let $H_{G}$ denote the \textsc{Quantum Max-$d$-Cut} Hamiltonian for $G$ then we have $H_{G'} = P_{\mathcal{H}'}(H_{G} \otimes \operatorname{id}_{(\mathbb{C}^{d})^{\otimes (k - 1)n}})P_{\mathcal{H}'}$.

Letting $x \in (S^{2d - 1})^{n}$ be the optimal solution of the rank-$(2d - 1)$ Grothendieck problem over $G$ and $z \in (\mathbb{C}P^{d})^{n}$ its image under the quotient map, we have that $\braket{c_{\mathcal{L}^{k},z} | H_{G'} | c_{\mathcal{L}^{k},z}}$ is equal to the value of $x$. Thus the largest eigenvalue of $H_{G'}$ is at least $C$.

For the other inequality, let $\rho$ be a state with maximal energy. Then $\rho$ is supported on $\mathcal{H}'$, so by Proposition~\ref{prop:multisym} the reduced state $\rho_{1}$ on the set of subsystems $V(G) \times \{1\}$ is $\epsilon$-close in trace distance to a mixture $\rho'$ over coherent states. Thus
\begin{align*}
|\operatorname{Tr}(H_{G'}\rho) - \operatorname{Tr}(H_{G}\rho')| &= |\operatorname{Tr}(P_{\mathcal{H}'}(H_{G}\otimes \operatorname{id}_{(\mathbb{C}^{d})^{\otimes (k - 1)n}})P_{\mathcal{H}'}\rho) - \operatorname{Tr}(H_{G}\rho')| \\
&= |\operatorname{Tr}(H_{G}\rho_{1}) - \operatorname{Tr}(H_{G}\rho')| \\
&\leq |H_{G}|_{\text{op}} |\rho_{1} - \rho'| \\
&\leq |H_{G}|_{\text{op}} \epsilon.
\end{align*}
Lastly, by the argument in the previous paragraph $\operatorname{Tr}(H_{G}\rho')$ is equal to the expected objective function value of the rank-$(2d - 1)$ Grothendieck problem with respect to the Husimi Q-distribution of $\rho$. Thus there is some feasible solution $x$ whose value is within $\epsilon$ of $\operatorname{Tr}(H_{G'}\rho)$, and the optimal solution also satisfies this.
\end{proof}

\begin{theorem}
The integrality gap for the level-1 ncSoS relaxation of \textsc{Quantum Max-$d$-Cut} is equal to the integrality gap for the level-1 SoS relaxation of the rank-$(2d - 1)$ Grothendieck problem.
\end{theorem}
\begin{proof}
Fix an integer $d \geq 2$. For $G$ a weighted graph, we denote by $\lambda_{0}(G)$ the maximum energy of the \textsc{Quantum Max-$d$-Cut} Hamiltonian of $G$, $ncSDP_{1}(G)$ the value of the optimal solution to the level-1 ncSoS SDP relaxation of \textsc{Quantum Max-$d$-Cut}, $OPT(G)$ the optimal solution of the rank-$(2d - 1)$ Grothendieck problem, and $SDP_{1}(G)$ the value of the optimal solution to the level-1 classical SoS SDP relaxation of the rank-$(2d - 1)$ Grothendieck problem.

Let $G$ be a weighted graph, $k \geq 0$ an integer, and $G' = K_{k} \times G$ be the cartesian product of $K_{k}$ with $G$, where all the weighted have been rescaled by $1/k^{2}$. We first claim that $SDP_{1}(G) = ncSDP_{1}(G')$. To see this, note that the level-1 SDP for the rank-$(2d - 1)$ Grothendieck problem is invariant under an $O(2d)$-action which simultaneously rotates the vector associated to each vertex. By Schur's lemma, the SDP solution is of the form $I_{2d} \otimes M$ for some matrix $M$. Similarly, the level-1 SDP for \textsc{Quantum Max-$d$-Cut} is invariant under a $U(d) \times S_{k}$-action where $U(d)$ acts on the Lie algebra $\mathfrak{u}(d)$ by the coadjoint action and $S_{k}$ acts by permuting each ``cloud.'' Thus the SDP solution is of the form $I_{(d^{2} - 1)k} \otimes M'$ for some matrix $M'$. For both SDPs the objective functions are of the form $C - \operatorname{Tr}(AM)$ and $C - \operatorname{Tr}(AM')$ respectively, where $A$ is the weighted adjacency matrix of $G$ and $C$ is the sum of all edge weights. These objective functions are equal, so in fact $M = M'$.

By Proposition~\ref{prop:maxdcut-approx}, we have that $|\lambda_{0}(G') - OPT(G)| \leq O(1/k)$. By taking $k \to \infty$ we see that the ratio $\lambda_{0}(G')/ncSDP_{1}(G')$ converges to $OPT(G)/SDP_{1}(G)$. Since the integrality gap of the SDP is an infimum over all instances and since this is true for all $G$, it follows that the integrality gap for \textsc{Quantum Max-$d$-Cut} is at most that of the rank-$(2d - 1)$ Grothendieck problem. To show the inequality in the other direction, let $G$ be a graph and $x \in (S^{2d - 1})^{n}$ be the optimal solution to the rank-$k$ Grothendieck problem on $G$. Let $z \in (\mathbb{C}P^{d})^{n}$ be the image of $x$ under the quotient and $c_{z}$ be the coherent state centered at $z$ for $\bigotimes \mathcal{L}_{i}$. Then $\braket{c_{z}|H_{G}|c_{z}}$ is equal to $OPT(G)$, so
\[ \frac{\lambda_{0}(G)}{ncSDP_{1}(G)} \geq \frac{OPT(G)}{SDP_{1}(G)}. \]
\end{proof}

The approximation of the rank-$(2d - 1)$ Grothendieck problem by \textsc{Quantum Max-$d$-Cut} also allows us to show hardness, assuming classical hardness of the Grothendieck problem.

\begin{proposition}
Suppose that it is $\mathsf{NP}$-hard to approximate the rank-$(2d - 1)$ Grothendieck problem to within $\alpha_{0}$. Then it is also $\mathsf{NP}$-hard to approximate \textsc{Quantum Max-$d$-Cut} to within any $\alpha > \alpha_{0}$.
\end{proposition}
\begin{proof}
Let $A \geq 1$ be an arbitrary integer. Consider the reduction from the rank-$(2d - 1)$ Grothendieck problem to \textsc{Quantum Max-$d$-Cut} which takes a graph $G$ on $n$ vertices and constructs a new graph $G'$ where each vertex is replaced with a cloud of size $An$ vertices, as defined in Proposition~\ref{prop:maxdcut-approx}. Then by Proposition~\ref{prop:multisym} and Proposition~\ref{prop:maxdcut-approx}, the maximum energy of $G'$ is at least $C$ and at most
\[ C + \frac{2nd}{An} = C + \frac{2d}{A}. \]
For any fixed $A$, the existence of this reduction shows that it is $\mathsf{NP}$-hard to approximate \textsc{Quantum Max-$d$-Cut} to within $\alpha + \frac{2d}{A}$. Taking $A$ to a sufficiently large constant then proves the statement.
\end{proof}


\subsection{The de Finetti theorem for representations of \texorpdfstring{$U(n)$}{U(n)}}
In \cite{Christandl_Koenig_Mitchison_Renner_2007}, a de Finetti theorem for irreducible representations of the unitary group $U(n)$ is shown. To fit this result into the framework of geometric quantization we will use the Borel-Weil theorem \cite{Baston_Eastwood_2016, Timchenko_2014, Snow} from geometric representation theory, which realizes the irreducible representations of $U(n)$ as spaces of global sections of holomorphic line bundles. We will not prove the theorem here but give some intuition relating it to the original quantum de Finetti theorem.

\begin{definition}
The \textbf{flag variety} $F_{n}$ (in type A) is the quotient $GL_{n}(\mathbb{C})/B_{n}(\mathbb{C})$, where $B_{n}$ is the group of upper-triangular matrices.
\end{definition}

\begin{proposition}
The flag variety is a smooth projective variety and K\"{a}hler manifold with a transitive action by $U(n)$. As a manifold, it can also be realized as $U(n)/U(1)^{n}$, where $U(1)^{n}$ is the group of diagonal unitary matrices.
\end{proposition}

\begin{theorem}[Borel-Weil]
For every partition $\lambda$ with at most $n$ rows, there is a line bundle $\mathcal{L}_{\lambda}$ over $F_{n}$ such that $\Gamma_{\text{hol}}(\mathcal{L}_{\lambda}) \cong V_{\lambda}$, where $V_{\lambda}$ is the irreducible representation of $U(n)$ indexed by $\lambda$. Moreover, $\mathcal{L}_{\lambda} \otimes \mathcal{L}_{\mu} \cong \mathcal{L}_{\lambda + \mu}$ for all $\lambda, \mu$, and $U(n)$ acts on $\mathcal{L}_{\lambda}$ by bundle automorphisms making the pair $(F_{n}, \mathcal{L}_{\lambda})$ homogeneous.
\end{theorem}

Applying Theorem~\ref{thm:main-thm} and Corollary~\ref{cor:homog-err}, we obtain the de Finetti theorem for representations of $U(n)$:

\begin{corollary}
If $\lambda$ and $\mu$ are partitions with at most $n$ rows, then for any state $\rho$ on $V_{\lambda + \mu}$ we have that $\operatorname{Tr}_{\mu} \rho$ is $\epsilon$-close to a mixture over coherent states for $\epsilon \leq 1 - \frac{\dim V_{\mu}}{\dim V_{\lambda + \mu}}$.
\end{corollary}

To give some intuition for how the line bundles $\mathcal{L}_{\lambda}$ are constructed, we will use an analogy with the circle bundle picture. Consider the (non-monotonic) partitions $e_{i}$ which have a 1 in the $i$-th position and zeroes everywhere else.
Constructing the vector bundle $\bigoplus _{i = 1} ^{n} \mathcal{L}_{e_{i}}$, we can take the corresponding $(S^{1})^{n}$-bundle (or torus bundle) $M$ over $F_{n}$. Decomposing $L^{2}_{\text{hol}}(M)$ into Fourier components using the $U(1)^{n}$-action, it can be shown that the component of Fourier weight $\lambda$ is isomorphic to $V_{\lambda}$.
It turns out that the torus bundle $M$ is just the unitary group $U(n)$ itself, and so sections of $\mathcal{L}_{\lambda}$ are just holomorphic function $f : U(n) \to \mathbb{C}$ of Fourier weight $\lambda$. This is equivalent to the condition that
\[ f(tu) = \chi_{\lambda}(t)f(u) \]
for all $t \in U(1)^{n}$ and $u \in U(n)$, where $\chi_{\lambda}$ is the 1-dimensional representation of $U(1)^{n}$ corresponding to $\lambda$. One can use this to prove the Borel-Weil theorem assuming known results from representation theory, or use the Borel-Weil theorem to give alternative proofs \cite{Timchenko_2014}.

Additionally, in some cases the introduction of the phase space can simplify some calculations. For example, for rectangular partitions of the form $\lambda = (k, \dots, k, 0, \dots, 0)$ the line bundles $V_{\lambda}$ are invariant under the subgroup $U(n_{1}) \times U(n_{2}) \subseteq U(n)$ and so induce line bundles over a Grassmannian manifold.
The corresponding coherent states will end up being tensor powers of Slater determinants, analogous to the usual case where the coherent states are tensor powers of pure states. In future work the author will describe some applications of this simplification.

\subsection{Renner's exponential de Finetti theorem}\label{sec:renner-exp}
Renner's exponential de Finetti theorem \cite[Theorem~4.3.2]{Renner_2005} (see also \cite{Vidick_Yuen_2016}, \cite[Theorem~8]{Harrow_2013}, or \cite[Corollary~V.2]{Koenig_Mitchison_2009}) approximates the marginal of a symmetric quantum state by a classical mixture over larger classes of states than the original de Finetti theorem, where the classes are indexed by an integer parameter $0 \leq r \leq k$. The resulting error bound has an exponential improvement in $r$ compared to the original quantum de Finetti theorem.

\begin{definition}
An \textbf{$r$-almost product state} on $\operatorname{Sym}^{k}(\mathbb{C}^{d})$ is a pure state which can be written as
\[ P_{\operatorname{Sym}^{k}(\mathbb{C}^{d})} (\ket{\psi} \otimes \ket{z}^{\otimes (k - r)}) \]
where $\ket{\psi} \in \operatorname{Sym}^{r}(\mathbb{C}^{d})$ and $\ket{z} \in \mathbb{C}^{d}$. We will also say that a state is an $(r, z)$-almost product state when the second tensor factor can be taken to be $z$.
\end{definition}

The coherent states in the exponential de Finetti theorem are the $r$-almost product states. In order to deduce the exponential de Finetti theorem from Theorem~\ref{thm:main-thm}, we will need to construct a phase space and line bundles such that the corresponding coherent states are $r$-almost product states.
To do this, we will use the generalization in Section~\ref{sec:quant-vecbund} of the construction of the Hilbert space to vector bundles.
Throughout, fix integers $0 \leq r \leq k$. Let $z_{0} \coloneqq e_{1} \in \mathbb{C}^{d}$ be a nonzero vector and $V_{z} \subseteq \operatorname{Sym}^{k}(\mathbb{C}^{d}) \subseteq (\mathbb{C}^{d})^{k}$ be the set of $(r, z)$-almost product states.
Let $G \coloneqq GL_{d}(\mathbb{C})$ and $P_{z_{0}} \subseteq G$ be the parabolic subgroup of matrices which fix $\operatorname{span}_{\mathbb{C}} \{z_{0}\}$. We let $V_{z}$ have the inner product given by restricting the one from $(\mathbb{C}^{d})^{\otimes k}$.

Note that $G/P_{z_{0}} \cong \mathbb{C}P^{d}$ and that $V_{z_{0}}$ is invariant under the action of $P_{z_{0}}$ on $(\mathbb{C}^{d})^{\otimes k}$. We will take the vector bundle over $\mathbb{C}P^{d}$ to be $\mathcal{E} \coloneqq (G \times _{P_{z_{0}}} V_{z_{0}})^{*}$, which is holomorphic and is homogeneous with respect to $G$.
As a smooth vector bundle, we can also construct $\mathcal{E}$ as $U(d) \times _{U(1) \times U(d - 1)} V_{z_{0}}$, where $U(1) \times U(d - 1)$ is the group of unitary matrices which fix $\operatorname{span}_{\mathbb{C}} \{z_{0}\}$. The fiber above a point $\ket{z}\bra{z} \in \mathbb{C}P^{d}$ is $V_{z}$,
and we take the Hermitian metric on $\mathcal{E}$ by restricting the constant Hermitian metric from the trivial bundles $\mathbb{C}P^{d} \times \operatorname{Sym}^{k}(\mathbb{C}^{d}) \subseteq \mathbb{C}P^{d} \times (\mathbb{C}^{d})^{\otimes k}$.

Following Section~\ref{sec:quant-vecbund}, we take $M$ to be the total space of the fiber bundle $p : \mathbb{P}(\mathcal{E}) \to \mathbb{C}P^{d}$ with the K\"{a}hler metric given fiberwise by the Fubini-Study metric.
We will take $\mathcal{L}_{1} \coloneqq \mathcal{O}_{\mathcal{E}}(1)$, which fiberwise is the anticanonical line bundle over each fiber, and $\mathcal{L}_{2} \cong p^{*} \mathcal{O}(1)$, which is the pullback along $p$ of the anticanonical line bundle over $\mathbb{C}P^{d}$. It is known \cite[Equation~A.2a]{Lazarsfeld_2004} that the pushforward bundle $p_{*}\mathcal{L}_{1}$ is isomorphic to $\mathcal{E}$.

\begin{remark}
The purpose of the line bundles over $M$ is to make it easier to prove Theorem~\ref{thm:main-thm}, which need a space indexing the coherent states in order to define a coherent state POVM. However, the error bound is in terms of the dimensions of the Hilbert spaces, and calculating those dimensions is easier to do in terms of vector bundles over $\mathbb{C}P^{d}$ instead of line bundles over $M$. The next two propositions formalize the needed facts.
\end{remark}

\begin{proposition}\label{prop:pushforward-iso}
We have $p_{*}(\mathcal{L}_{1} \otimes \mathcal{L}_{2}^{k}) \cong \mathcal{E} \otimes \mathcal{O}(1)^{k}$.
\end{proposition}
\begin{proof}
First recall that the pushforward functor $p_{*}$ is the right adjoint of the inverse image functor $p^{*}$. The adjunction map $\operatorname{Hom}(p^{*}A, B) \cong \operatorname{Hom}(A, p_{*}B)$ maps $f : p^{*}A \to B$ to the composition of $p_{*}f$ with the unit map $\eta_{A} : A \mapsto p_{*}(p^{*}A)$.
For the particular map $p$ we see that fiberwise over $\mathbb{C}P^{d}$ the bundle $p^{*}A$ is trivial. Thus the pushforward $p_{*}(p^{*}A)$ is isomorphic on stalks to $A$ and thus isomorphic to $A$, and it can be checked that the unit of the adjunction is this isomorphism. Thus the adjunction map sends isomorphisms to isomorphisms.

We have
\begin{align*}
\operatorname{Hom}(\mathcal{E} \otimes \mathcal{O}(1)^{k}, p_{*}(\mathcal{L}_{1} \otimes \mathcal{L}_{2}^{k})) &\cong \operatorname{Hom}(p^{*}(\mathcal{E} \otimes \mathcal{O}(1)^{k}), \mathcal{L}_{1} \otimes \mathcal{L}_{2}^{k}) \\
&\cong \operatorname{Hom}((p^{*}\mathcal{E}) \otimes \mathcal{L}_{2}^{k}), \mathcal{L}_{1} \otimes \mathcal{L}_{2}^{k}).
\end{align*}
Since we have an isomorphism $p_{*}\mathcal{L}_{1} \cong \mathcal{E}$, the argument from the previous paragraph shows that there is an isomorphism $p^{*}\mathcal{E} \cong \mathcal{L}_{1}$.
By tensoring we get an isomorphism $p^{*}\mathcal{E} \otimes \mathcal{L}_{2}^{k} \cong \mathcal{L}_{1} \otimes \mathcal{L}_{2}^{k}$. Using the argument from the previous paragraph, the adjunction map sends this to an isomorphism $\mathcal{E} \otimes \mathcal{O}(1)^{k} \cong p_{*}(\mathcal{L}_{1} \otimes \mathcal{L}_{2}^{k})$.
\end{proof}



Some intuition for the following proposition can be given based on the usual proof of the original de Finetti theorem for $\mathbb{C}P^{d}$.
The canonical line bundle $\mathcal{O}(-1)$ can be thought of as a subbundle of the trivial bundle $\mathbb{C}P^{d} \times \mathbb{C}^{d}$ where the fiber over $x$ is the corresponding 1-dimensional subspace of $\mathbb{C}^{d}$.
The higher tensor powers $\mathcal{O}(-1)^{k}$ are then subspaces of $(\mathbb{C}^{d})^{k}$ whose fibers are the 1-dimensional subspaces spanned by $x^{\otimes k}$.

These subspaces are just the coherent states corresponding to the points $x$. Since global holomorphic sections of the trivial bundle $\mathbb{C}P^{d} \times (\mathbb{C}^{d})^{\otimes k}$ are elements of $(\mathbb{C}^{d})^{\otimes k}$, an appropriate (essentially bookkeeping) argument involving the duals of the bundles involved shows that the space of global holomorphic sections of $\mathcal{O}(-1)^{k}$ is the subspace of $(\mathbb{C}^{d})^{\otimes k}$ spanned by the coherent states $x^{\otimes k}$.
One could then use the usual representation-theoretic argument to determine that the tensor power states span the symmetric subspace, instead of using the argument we gave before. In the next two propositions we essentially carry out the representation-theoretic argument for our vector bundle over $\mathbb{C}P^{d}$, since a more general theorem does not seem to be known.

\begin{proposition}\label{prop:hs-rep}
The space of global holomorphic sections of $\mathcal{L}_{1} \otimes \mathcal{L}_{2}^{\otimes (n - k)}$ is the smallest $U(d)$-subrepresentation of $\operatorname{Sym}^{k}(\mathbb{C}^{d}) \otimes \operatorname{Sym}^{n - k}(\mathbb{C}^{d})$ which contains $\mathcal{E}_{z_{0}} \otimes z_{0}^{\otimes (n - k)}$.
\end{proposition}
\begin{proof}
By Proposition~\ref{prop:pushforward-iso}, the space of global holomorphic sections of $\mathcal{L}_{1} \otimes \mathcal{L}_{2}^{n - k}$ over $M$ is isomorphic to the space of global holomorphic sections of $\mathcal{E} \otimes \mathcal{O}(1)^{n - k}$ over $\mathbb{C}P^{d}$. The dual bundle $\mathcal{E}^{*} \otimes \mathcal{O}(-1)^{k}$ is a holomorphic subbundle of a trivial bundle. Let the inclusion map be
\[ \iota : \mathcal{E}^{*} \otimes \mathcal{O}(-1)^{k} \hookrightarrow \mathbb{C}P^{d} \times (\operatorname{Sym}^{k}(\mathbb{C}^{d}) \otimes \operatorname{Sym}^{n - k}(\mathbb{C}^{d})). \]
Dualizing, we get a surjection
\[ \mathbb{C}P^{d} \times (\operatorname{Sym}^{k}(\mathbb{C}^{d}) \otimes \operatorname{Sym}^{n - k}(\mathbb{C}^{d}))^{*} \twoheadrightarrow \mathcal{E} \otimes \mathcal{O}(1)^{k}. \]

We claim that the induced map on the spaces of global holomorphic sections of the two bundles is also a surjection. Letting $A$ be the first (trivial) bundle and $B$ the second, we have a long exact sequence
\[ 0 \rightarrow H^{0}(A/B) \rightarrow H^{0}(A) \rightarrow H^{0}(B) \rightarrow H^{1}(A/B) \rightarrow \cdots \]
so it suffices to show that $H^{1}(A/B) = 0$. Note that there is an isomorphism
\[ A \cong G \times_{H} (\operatorname{Sym}^{k}(\mathbb{C}^{d}) \otimes \operatorname{Sym}^{n - k}(\mathbb{C}^{d}))^{*} \]
so $A$ is also a homogeneous vector bundle over $\mathbb{C}P^{d}$. Then we have
\[ A/B \cong G \times_{H} (\operatorname{Sym}^{k}(\mathbb{C}^{d}) \otimes \operatorname{Sym}^{n - k}(\mathbb{C}^{d}))^{*}/(\mathcal{E}_{z_{0}} \otimes z_{0}^{\otimes (n - k)}) \]
so $A/B$ is a homogeneous vector bundle. The Borel-Weil-Bott theorem \cite[Proposition~11.4]{Snow} then implies that $H^{1}(A/B) = 0$, since the weights appearing in a fiber of $A/B$ are a subset of the weights appearing in a fiber of $A$.

Finally, we calculate the kernel of the surjection $H^{0}(A) \twoheadrightarrow H^{0}(B)$. Fiberwise, the maps corresponds to restricting a linear function $\operatorname{Sym}^{k}(\mathbb{C}^{d}) \otimes \operatorname{Sym}^{n - k}(\mathbb{C}^{d}) \to \mathbb{C}$ to a linear function $\mathcal{E}_{x} \otimes x^{\otimes (n - k)} \to \mathbb{C}$.
Thus the kernel is precisely all linear functions on $\operatorname{Sym}^{k}(\mathbb{C}^{d}) \otimes \operatorname{Sym}^{n - k}(\mathbb{C}^{d})$ which vanish on the subspaces $\mathcal{E}_{x} \otimes x^{\otimes (n - k)}$ for all unit vectors $x \in \mathbb{C}^{d}$.
Taking duals and using the fact that the irreducible representations of $U(d)$ gives the desired statement.
\end{proof}


\begin{proposition}\label{prop:exp-reps}
The space of global holomorphic sections of $\mathcal{L}_{1} \otimes \mathcal{L}_{2}^{\otimes (n - k)}$ is a $U(d)$-representation which is isomorphic to
\[ \bigoplus _{i = 0} ^{r} S_{(n - i, i)}(\mathbb{C}^{d}) \]
where $S_{\mu}(\mathbb{C}^{d})$ denotes the irreducible representation of $\mathbb{C}^{d}$ indexed by the partition $\mu$.
\end{proposition}
\begin{proof}
By Proposition~\ref{prop:hs-rep}, we know that the Hilbert space is
\[ \mathcal{H} \cong \operatorname{span}_{\mathbb{C}} \bigcup _{U \in U(d)} U \cdot (\mathcal{E}_{z_{0}} \otimes z_{0}^{\otimes (n - k)}) \]
and is contained in $\operatorname{Sym}^{k}(\mathbb{C}^{d}) \otimes \operatorname{Sym}^{n - k}(\mathbb{C}^{d})$. By Pieri's formula (see \cite[Theorem~40.4]{Bump_2013} or \cite[Equation~5.16]{Macdonald_1999}), we have
\begin{align*}
\operatorname{Sym}^{k}(\mathbb{C}^{d}) \otimes \operatorname{Sym}^{n - k}(\mathbb{C}^{d}) &\cong \operatorname{Sym}^{k}(\mathbb{C}^{d}) \otimes S_{(n - k)}(\mathbb{C}^{d}) \\
&\cong \bigoplus _{\lambda} S_{\lambda}(\mathbb{C}^{d})
\shortintertext{where $\lambda$ has $k$ additional boxes with no two in the same column, which is then}
&\cong \bigoplus _{i = 0} ^{k} S_{(n - k + i, k - i)}(\mathbb{C}^{d}).
\end{align*}

By Frobenius reciprocity, there is a $U(d)$-equivariant surjection
\[ \operatorname{Ind}_{U(1) \times U(d - 1)}^{U(d)} (\mathcal{E}_{z_{0}} \otimes z_{0}^{\otimes (n - k)}) \twoheadrightarrow \mathcal{H} \]
from the induced representation to $\mathcal{H}$. As a $(U(1) \times U(d - 1))$-representation, we have
\[ (\mathcal{E}_{z_{0}} \otimes z_{0}^{\otimes (n - k)}) \cong \bigoplus _{i = 0} ^{r} \operatorname{Sym}^{i}(\mathbb{C}^{d - 1}) \otimes W^{\otimes (n - i)} \]
where $W$ is the one-dimensional representation of $U(1)$.
It is known from the branching rule from $U(d)$ to $U(1) \times U(d - 1)$ \cite[Theorem~41.1]{Bump_2013} that the induced representation of an irreducible representation $W^{\otimes k} \otimes S_{\lambda}(\mathbb{C}^{d - 1})$ of $U(1) \times U(d - 1)$ contains $S_{\mu}(\mathbb{C}^{d})$ once for every $\mu$ such that $|\mu| = k + |\lambda|$ and $\mu_{1} \geq \lambda_{1} \geq \mu_{2} \geq \cdots \geq \lambda_{d - 1} \geq \mu_{d}$.
\end{proof}

\begin{proposition}\label{prop:hookcontent-app}
We have
\[ \dim S_{(n - i, i)}(\mathbb{C}^{d}) = \frac{n - 2i + 1}{n - i + 1}\binom{n - i + d - 1}{d - 1}\binom{i + d - 2}{d - 2}. \]
\end{proposition}
\begin{proof}
We use Stanley's hook-content formula \cite[Theorem~7.21.2]{Stanley_2001} to calculate the number of semistandard Young tableaux of shape $(n - i, i)$ with entries bounded by $d$, which is then equal to the dimension of $S_{(n - i, i)}(\mathbb{C}^{d})$ \cite[Proposition~41.2]{Bump_2013}. In the first row the content of a cell in column $c$ is $c - 1$, and in the second row the content of a cell in column $c$ is $c - 2$.
For a cell in the first row and column $c \geq i$, the hook length is $n - r - c + 1$. For a cell in the first row and column $c \leq i$, the hook length is $n - r - c + 2$. For a cell in the second row and column $c$, the hook length is $r - c + 1$. Overall, the denominator in the hook-content formula is the product of the hook lengths, which is
\begin{align*}
\left(\prod _{j = n - 2i + 2} ^{n - r + 1} j\right)\left(\prod _{j = 1} ^{n - 2r} j\right)\left(\prod _{j = 1} ^{i} j\right) &= \frac{(n - i + 1)!}{(n - 2i + 1)!}(n - 2i)!i!.
\end{align*}
The numerator, using the hook-contents, is
\begin{align*}
\left(\prod _{j = 0} ^{n - i - 1} (d + i)\right)\left(\prod _{j = -1} ^{i - 2} (d + j)\right) &= \left(\frac{(d + n - i - 1)!}{(d - 1)!}\right)\left(\frac{(d + i - 2)!}{(d - 2)!}\right).
\end{align*}
Taking the ratio gives
\begin{align*}
\dim S_{(n - i, i)}(\mathbb{C}^{d}) &= \frac{(n - i + d - 1)!(i + d - 2)!}{(n - i + 1)!(n - 2i + 1)!^{-1}(n - 2i)!i!(d - 1)!(d - 2)!} \\
&= (n - 2i + 1)\frac{(n - i + d - 1)!(i + d - 2)!}{(n - i + 1)!i!(d - 1)!(d - 2)!} \\
&= \frac{n - 2i + 1}{n - i + 1}\left(\frac{(n - i + d - 1)!}{(n - i)!(d - 1)!}\right)\left(\frac{(i + d - 2)!}{i!(d - 2)!}\right) \\
&= \frac{n - 2i + 1}{n - i + 1}\binom{n - i + d - 1}{d - 1}\binom{i + d - 2}{d - 2}.
\end{align*}
\end{proof}

\begin{proposition}\label{prop:exp-dim}
We have
\[ \dim \mathcal{H}_{\mathcal{L}_{1} \otimes \mathcal{L}_{2}^{\otimes (n - k)}} = \binom{n + d - r - 1}{d - 1}\binom{d + r - 1}{d - 1}. \]
\end{proposition}
\begin{proof}
From Proposition~\ref{prop:exp-reps} and Proposition~\ref{prop:hookcontent-app}, we have
\begin{align*}
\dim \mathcal{H}_{\mathcal{L}_{1} \otimes \mathcal{L}_{2}^{\otimes (n - k)}} &= \sum _{i = 0} ^{r} \dim S_{(n - i, i)}(\mathbb{C}^{d}) \\
&= \sum _{i = 0} ^{r} \frac{n - 2i + 1}{n - i + 1}\binom{n - i + d - 1}{d - 1}\binom{i + d - 2}{d - 2}.
\end{align*}
We then prove the statement by induction on $r$. For the base case $r = 0$, we have
\begin{align*}
\binom{n + d - r - 1}{d - 1}\binom{d + r - 1}{d - 1} &= \binom{n + d - 1}{d - 1} \\
&= \dim \operatorname{Sym}^{n}(\mathbb{C}^{d}) \\
&= \dim S_{(n)}(\mathbb{C}^{d}).
\end{align*}
For the induction step, we have
\begin{align*}
&\phantom{={}} \binom{n + d - r - 1}{d - 1}\binom{d + r - 1}{d - 1} + \frac{n - 2(r + 1) + 1}{n - (r + 1) + 1}\binom{n - (r + 1) + d - 1}{d - 1}\binom{r + 1 + d - 2}{d - 2} \\
&= \binom{n + d - r - 1}{d - 1}\binom{d + r - 1}{d - 1} + \frac{n - 2r - 1}{n - r}\binom{n + d - r - 2}{d - 1}\binom{d + r - 1}{d - 2} \\
&= \frac{(n + d - r - 2)!}{(d - 1)!(n - r - 1)!} \cdot \frac{(d + r - 1)!}{r!(d - 2)!}\left( \frac{n + d - r - 1}{n - r}\cdot\frac{1}{d - 1} + \frac{n - 2r - 1}{n - r}\cdot\frac{1}{r + 1} \right) \\
&= \frac{(n + d - r - 2)!}{(d - 1)!(n - r)!} \cdot \frac{(d + r - 1)!}{r!(d - 2)!}\left(\frac{n + d - r - 1}{d - 1} + \frac{n - 2r - 1}{r + 1} \right) \\
&= \frac{(n + d - r - 2)!(d + r - 1)!}{(d - 1)!(n - r)!r!(d - 2)!} \cdot \frac{(n - r)(d - 1 + r + 1) + (r + 1)(d - 1) + (d - 1)(-r - 1)}{(d - 1)(r + 1)} \\
&= \frac{(n + d - r - 2)!(d + r - 1)!}{(d - 1)!(n - r)!r!(d - 2)!} \cdot \frac{(n - r)(d + r)}{(d - 1)(r + 1)} \\
&= \frac{(n + d - r - 2)!}{(d - 1)!(n - r - 1)!} \cdot \frac{(d + r)!}{(r + 1)!(d - 1)!} \\
&= \binom{n + d - r - 2}{d - 1}\binom{d + r}{d - 1}
\end{align*}
completing the induction.
\end{proof}

\begin{theorem}
Any state in $\mathcal{H}_{\mathcal{L}_{1} \otimes \mathcal{L}_{2}^{\otimes (n - k)}}$ is $\epsilon$-close to a convex combination of $r$-almost product states, where
\[
\epsilon \leq 1 - \left(1 - \frac{2dk}{n}\right)\left(1 + \frac{d - 1}{n}\right)^{r}.
\]
\end{theorem}
\begin{proof}
Consider the group $U(d)$ acting on $M$ by bundle automorphisms preserving the Hermitian metrics on $\mathcal{L}_{1}$ and $\mathcal{L}_{2}$. Under the projection $p : M \to \mathbb{C}P^{d}$ the $U(d)$-action become the usual action on $\mathbb{C}P^{d}$. Recall that we have chosen the inner products on the Hilbert spaces $\mathcal{H}_{\mathcal{L}_{1} \otimes \mathcal{L}_{2}^{\otimes (n - k)}}$ to be the restrictions of the inner products on $(\mathbb{C}^{d})^{\otimes n}$. Thus the unitary group $U(V_{z_{0}})$ acts on $M$, again preserving the Hermitian metrics, fiberwise on each fiber of the projection over $\mathbb{C}P^{d}$. We take the measure on $M$ to be the volume measure of the Fubini-Study metric constructed from the Hermitian metric on $\mathcal{E}$, which is invariant under both the $U(d - 1)$- and $U(V_{z_{0}})$-actions.

The decomposition of $\mathcal{H}_{\mathcal{L}_{1} \otimes \mathcal{L}_{2}^{\otimes (n - k)}}$ is multiplicity-free by Proposition~\ref{prop:exp-reps}. By a similar argument as in Corollary~\ref{cor:homog-err}, integrating the coherent state POVM over $M$ will give a self-adjoint operator which is a scalar multiple of the identity on each irreducible representation contained in $\mathcal{H}_{\mathcal{L}_{1} \otimes \mathcal{L}_{2}^{\otimes (n - k)}}$. Restricting to a representation of $U(1) \times U(d - 1)$, the fact that the Hilbert space is a representation of $U(V_{z_{0}})$ implies that each of these scalars must be equal and the integral of the POVM is a scalar multiple of the identity.

To determine the scalar, we write the integral of the POVM as
\begin{align*}
\int _{M} c_{z}\,d\mu(z) &= \int _{\mathbb{C}P^{d}} \left(\int _{p^{-1}(z)} c_{w}\,d\mu_{p^{-1}(z)}(w)\right)\,d\mu_{FS}(z) \\
&= \int _{\mathbb{C}P^{d}} P_{V_{z}}\,d\mu_{FS}(z)
\end{align*}
where $P_{V_{z}}$ is the projector onto the subspace $V_{z} \subseteq \mathcal{H}_{\mathcal{L}_{1} \otimes \mathcal{L}_{2}^{\otimes (n - k)}}$. The second equality holds due to the choice of the Hermitian metrics and inner products, which are such that the coherent state $c_{z}$ is equal to the corresponding $r$-almost product state and the integral gives the projector onto $V_{z}$. Thus the correct normalization factor is $(\dim \mathcal{H}_{\mathcal{L}_{1} \otimes \mathcal{L}_{2}^{\otimes (n - k)}})/(\dim V_{z_{0}})$. We have
\begin{align*}
\dim V_{z_{0}} &= \sum _{i = 0} ^{r} \dim \operatorname{Sym}^{i}(\mathbb{C}^{d - 1}) \\
&= \dim \operatorname{Sym}^{r}(\mathbb{C}^{d}) \\
&= \binom{d + r - 1}{d - 1}
\end{align*}
so by Proposition~\ref{prop:exp-dim} the normalization factor is
\begin{align*}
\frac{\dim \mathcal{H}_{\mathcal{L}_{1} \otimes \mathcal{L}_{2}^{\otimes (n - k)}}}{\dim V_{z_{0}}} &= \frac{\binom{n + d - r - 1}{d - 1}\binom{d + r - 1}{d - 1}}{\binom{d + r - 1}{d - 1}} \\
&= \binom{n + d - r - 1}{d - 1}.
\end{align*}
We have
\[
\dim \mathcal{H}_{\mathcal{L}_{2}^{\otimes (n - k)}} = \dim \operatorname{Sym}^{n - k}(\mathbb{C}^{d}) = \binom{n - k + d - 1}{d - 1}
\]
and
\begin{align*}
\binom{n + d - 1}{n} &= \binom{n - r + d - 1}{n - r}\prod _{i = 0} ^{r - 1} \frac{n - i + d - 1}{n - i} \\
&= \binom{n - r + d - 1}{d - 1}\prod _{i = 0} ^{r - 1} \frac{n - i + d - 1}{n - i}.
\end{align*}
Thus
\begin{align*}
&\phantom{{}={}} \binom{n - k + d - 1}{d - 1}\binom{n + d - 1}{d - 1}^{-1} \prod _{i = 0} ^{r - 1} \frac{n - i + d - 1}{n - i} \\
&\geq \binom{n - k + d - 1}{d - 1}\binom{n + d - 1}{d - 1}^{-1} \prod _{i = 0} ^{r - 1} \frac{n + d - 1}{n} \\
&= \binom{n - k + d - 1}{d - 1}\binom{n + d - 1}{d - 1}^{-1} \left(1 + \frac{d - 1}{n}\right)^{r} \\
&\geq \left(1 - \frac{2dk}{n}\right)\left(1 + \frac{d - 1}{n}\right)^{r}
\end{align*}
applying the proof of Proposition~\ref{prop:origdf} in the last line, giving the final error bound by Theorem~\ref{thm:main-thm}.
\end{proof}

\subsection{Error bounds for classical optimization}
One of application of the original quantum de Finetti theorem is for error bounds for certain classes separability testing problems. In general, in separability testing one is given a density matrix $\rho$ on a bipartite Hilbert space $\mathcal{H}_{A} \otimes \mathcal{H}_{B}$ and wants to compute the value of the optimization problem
\[ \max _{\ket{\psi}_{A}, \ket{\phi}_{B}} \operatorname{Tr}((\psi \otimes \phi)\rho). \]
Since states of the form $\psi \otimes \phi$ are the extreme points of the convex set of separable states, this quantity is roughly the distance of $\rho$ to the set of separable states. One can recover these results using the multi-symmetric quantum de Finetti theorem of Section~\ref{sec:multisym}.

To see the general strategy for deducing an error bound from a de Finetti theorem, suppose $M$ is a compact K\"{a}hler manifold and $f : M \to \mathbb{R}$ is a Hermitian polynomial of degree-$k$ on $M$. Concretely, this means that there is some line bundle $\mathcal{L}$ over $M$ and an operator $A : \mathcal{H}_{\mathcal{L}} \to \mathcal{H}_{\mathcal{L}}$ such that $f(z) = \braket{c_{z}|A|c_{z}}$ for any point $z \in M$ and corresponding coherent state $\ket{c_{z}}$. If $\mathcal{L}_{2}$ is some other line bundle on $M$, then the de Finetti theorem in general gives us a bound of $\epsilon$ on the trace distance between some state $\rho$ and the mixture
\[ \rho' \coloneqq \int _{M} c_{\mathcal{L},z}\operatorname{Tr}(c_{\mathcal{L}',z}\rho)\,d\mu(z) \]
over coherent states.
Now take $\rho$ to be the eigenstate of the operator $A$ with the largest eigenvalue. Since the trace distance is dual to the operator norm, this gives us a bound
\begin{align*}
\epsilon &\geq |\operatorname{Tr}(A\rho) - \operatorname{Tr}(A\rho')| \\
&= \left| |A|_{\text{op}} - \int _{M} \operatorname{Tr}(A c_{\mathcal{L},z}) \operatorname{Tr}(c_{\mathcal{L}',z}\rho)\,d\mu(z) \right| \\
&= \left| |A|_{\text{op}} - \int _{M} f(z) \operatorname{Tr}(c_{\mathcal{L}',z}\rho)\,d\mu(z) \right|.
\end{align*}
Notice that the integral is the expectation value of $f(z)$ with respect to the Husimi quasiprobability distribution of the state $\rho$, which is nonnegative and integrates to 1. Thus
\[
\sup _{z \in M} f(z) \geq \int _{M} f(z) \operatorname{Tr}(c_{\mathcal{L}',z}\rho)\,d\mu(z).
\]
We also have
\[ |A|_{\text{op}} \geq \sup _{z \in M} f(z) \]
since $\braket{c_{z}|A|c_{z}} = f(z)$, so computing $|A|_{\text{op}}$ is a relaxation of computing the maximum of $f$. Combining the two inequalities shows that the bound given by the de Finetti theorem is also a bound on how far $|A|_{\text{op}}$ can be from the true optimum of $f$.

By taking powers of the same line bundle $\mathcal{L}$, we get a hierarchy of spectral algorithms for optimizing Hermitian polynomial functions. Moreover, it turns out that this hierarchy is the same as the well-known Hermitian sum-of-squares (HSoS) hierarchy which has previously been studied in optimization and real algebraic geometry \cite{Putinar_2012, Catlin_DAngelo_1996, Catlin_DAngelo_1997, DAngelo_2005, DAngelo_2011, Johnston_Lovitz_Vijayaraghavan_2023}. In the case of $\mathbb{C}P^{d}$, the hierarchy can be defined as
\begin{align*}
& \min |A|_{\text{op}} \\
& \operatorname{st} \begin{aligned}[t]
\braket{z^{\otimes n}| A| z^{\otimes n}} = f(z)\,\forall z \in \mathbb{C}P^{d} \\
A = A^{*} \in (\mathbb{C}^{d \times d})^{\otimes n}
\end{aligned}
\end{align*}
where $f$ is a degree-$2k$ bihomogeneous Hermitian polynomial. For optimization of real polynomials over spheres, the traditional SoS hierarchy is defined in the same way except that $A$ is only required to be a symmetric matrix and $f$ can be any degree-$k$ homogeneous polynomial with real coefficients.

In the real case the matrix $A$ is not uniquely specified, as can be seen from the degree-4 polynomial $x_{1}x_{2}x_{3}x_{4}$. In the Hermitian case, one can check that the constraints uniquely specify $A$: since the state $\bra{z^{\otimes n}}$ are in the symmetric subspace $\operatorname{Sym}^{n}(\mathbb{C}^{d})$ we can assume that $A$ is supported on the symmetric subspace, and since $f$ is Hermitian the equality constraint uniquely specifies the matrix $A$. This matrix $A$ must then be the quantization of $f$ since the state $\ket{z^{\otimes n}}$ is the coherent state at the point $z \in M$, from which the equality constraint can be rewritten as $\operatorname{Tr}(c_{z}A) = f(z)$.

\subsubsection{Optimization of sparse polynomials using creation and annihilation operators}\label{sec:sparse}
An immediate benefit of the quantization perspective is that it can allow us to write the matrix $A$ in terms of creation and annihilation operators (as was also done in \cite{Lewin_Nam_Rougerie_2015} and \cite{Rougerie_2020}), which can be more efficient computationally than writing $A$ as a dense matrix as has been done previously \cite{Johnston_Lovitz_Vijayaraghavan_2023}. To see how creation and annihilation operators arise, first recall that for $\mathbb{C}P^{d}$ the quantized Hilbert space for the bundle $\mathcal{L}^{n}$ is isomorphic to the symmetric subspace $\operatorname{Sym}^{n}(\mathbb{C}^{d}$. For the multi-mode Segal-Bargmann space, which is the quantization of $\mathbb{C}^{d}$ and is the Hilbert space of $d$ bosonic modes, we have the isomorphism
\[ \mathcal{H}_{SB} \cong \bigoplus _{n = 0} ^{\infty} \operatorname{Sym}^{n}(\mathbb{C}^{d}) \]
which splits $\mathcal{H}_{SB}$ into the direct sum of $n$-particle subspaces. The coherent state corresponding to a point $z \in \mathbb{C}^{d}$ is (up to a normalization constant) $w \mapsto \exp(\bar{z} \cdot w)$ in the holomorphic representation, which under the isomorphism corresponds to the concatenation
$(1, z/1!, z^{\otimes 2}/2!, \dots)$ of the coherent states for $\mathbb{C}P^{d}$.

The annihilation operator $a_{i}$ for the mode $i$ splits as a direct sum of operators $\operatorname{Sym}^{n}(\mathbb{C}^{d}) \to \operatorname{Sym}^{n - 1}(\mathbb{C}^{d})$, and similarly for the creation operators $a_{i}^{*}$. If
\[ f = \sum _{\substack{|\alpha| = k\\|\beta| = k}} c_{\alpha\beta} \bar{z}^{\alpha}z^{\beta} \]
is a bihomogeneous Hermitian polynomial, then define
\[ A \coloneqq \sum _{\substack{|\alpha| = k\\|\beta| = k}} c_{\alpha\beta} (a^{*})^{\alpha}a^{\beta} \]
where $a = (a_{1}, \dots, a_{d})$ is a vector of the annihilation operators. From the splitting of the $a_{i}$ we see that $A : \mathcal{H}_{SB} \to \mathcal{H}_{SB}$ splits as a direct sum of operators $A_{n} : \operatorname{Sym}^{n}(\mathbb{C}^{d}) \to \operatorname{Sym}^{n}(\mathbb{C}^{d})$. The identity $a_{i}c_{SB,z} = z_{i}c_{SB,z}$ for the bosonic coherent states, along with the expression for the $n$-particle components of the coherent states, shows that
\[ \braket{z^{\otimes n}|A_{n}|z^{\otimes n}} = f(z) \]
for all $z \in \mathbb{C}P^{d}$, so $A_{n}$ is indeed the quantization of $f$.

Computationally, this gives a more efficient algorithm when the polynomial $f$ is sparse. In order to compute the largest eigenvalue of $A$ one typically constructs $A$ as a dense matrix and uses the power method, which involves repeatedly multiplying a vector by $A$. If $f$ is sparse then one can instead use an implicit representation of $A$ in terms of creation and annihilation operators, for which a matrix-vector multiplication would use only $2k\cdot nnz(f)$ applications of a creation or annihilation operators, where $nnz(f)$ is the number of nonzero coefficients of $f$.

\subsubsection{de Finetti theorems and bounds for the HSoS hierarchy on smooth projective varieties}
One can also generalize the results of the previous section to the setting where $M$ is a more general projective variety, assuming only that they are smooth and not that there is a transitive action by a compact group. One issue that arises when there is no such transitive action is that one cannot use arguments based on Schur's lemma to appropriately normalize the coherent state POVM. Instead we will prove the following proposition, which shows that a measure which produces an approximately normalized POVM still suffices to prove a de Finetti theorem.

\begin{proposition}\label{prop:approx-povm}
In the setup of Theorem~\ref{thm:main-thm}, suppose $\mu_{\mathcal{L}_{2}}$ and $\mu_{\mathcal{L}_{1} \otimes \mathcal{L}_{2}}$ are such that they define approximate coherent state POVMs with incompleteness bounded by $\delta$. That is,
\begin{align*}
\left| \int _{M} e_{\mathcal{L}_{2},z}\,d\mu_{\mathcal{L}_{2}}(z) - \operatorname{id}_{\mathcal{H}_{\mathcal{L}_{2}}} \right|_{1} &\leq \delta \\
\shortintertext{and}
\left| \int _{M} e_{\mathcal{L}_{1} \otimes \mathcal{L}_{2},z}\,d\mu_{\mathcal{L}_{2}}(z) - \operatorname{id}_{\mathcal{H}_{\mathcal{L}_{1} \otimes \mathcal{L}_{2}}} \right|_{1} &\leq \delta.
\end{align*}
Then the statement of Theorem~\ref{thm:main-thm} holds for a $\delta$-approximately normalized mixture distribution with error $\epsilon \leq 2R + 4\delta$ where
\[ R = \sup _{x \in M} \left(1 - \frac{d\mu_{\mathcal{L}_{2}}}{d \mu_{\mathcal{L}_{1} \otimes \mathcal{L}_{2}}}(x) \right) \]
as before. Specifically, if $\rho_{12}$ is a mixed state on $\mathcal{H}_{\mathcal{L}_{1} \otimes \mathcal{L}_{2}}$ and $\rho_{1}$ is the reduced density matrix of $M^{*}_{\mathcal{L}_{1},\mathcal{L}_{2}}\rho_{12}M_{\mathcal{L}_{1},\mathcal{L}_{2}}$ on $\mathcal{H}_{\mathcal{L}_{1}}$, then $\rho_{1}$ is $\epsilon$-close in trace distance to a mixture over coherent states.
\end{proposition}
\begin{proof}
The proof is the same as the proof of Theorem~\ref{thm:main-thm}, except that an additional additive error of $\delta$ is incurred in each steps where one of the POVMs is used.
This first occurs in Equation~\ref{eq:mainthm-povm1}, for which the equality will no longer hold exactly but only up to $\delta$ in trace distance, showing that Equation~\ref{eq:mainthm-povm2} holds with an additional term of $\delta$ on the right-hand side.
Similarly, an additional term of $\delta$ will be added to the right-hand side of Equation~\ref{eq:mainthm-povm3} where the POVM for $\mu_{\mathcal{L}_{1} \otimes \mathcal{L}_{2}}$ is used, eventually giving a slightly weaker bound of $|\alpha| \leq R/2 + \delta$.
The analogous argument for $\beta$ will result in a bound $|\beta| \leq R/2 + \delta$, and similarly in Equation~\ref{eq:mainthm-povm4} we will get a bound of $|\gamma| \leq R/2 + \delta$. Summing gives an overall bound of $2R + 4\delta$.
\end{proof}

Now suppose $M$ has some embedding $\iota : M \to \mathbb{P}(V)$ into some complex projective space $\mathbb{P}(V)$ equipped with the Fubini-Study metric associated to some Hermitian inner product on $V$ and $\mathcal{L} = \iota^{*}\mathcal{O}(1)$. Let the Hermitian metric $h$ on $\mathcal{L}$ be the pullback of the Fubini-Study metric on $\mathcal{O}(1)$. The inner products on $\mathcal{H}_{\mathcal{L}^{k}}$ are the restrictions of the ones on $\Gamma_{\text{hol}}(\mathbb{P}(V), \mathcal{L}^{k})$, which are given as before by integration with respect to the Fubini-Study metric.

\begin{example}
There can be a transitive group action on $M$ even if the metric $h$ on $\mathcal{L}$ is not invariant under the action. For example, csondier $M = \mathbb{C}P^{1}$ and we consider the projective embedding $[x : y] \mapsto [x^{2} : axy : y^{2}]$ for some constant $a$. If $a \neq \sqrt{2}$ then the $U(2)$-action does not preserve the metric on $\mathcal{O}(2)$ pulled back from $\mathbb{C}P^{2}$, while if $a = \sqrt{2}$ then the metric is a rescaling of the one pulled back from $\mathcal{O}(2)$.
\end{example}

Consider the sequence of operators
\[ A_{k} \coloneqq (\dim \mathcal{H}_{\mathcal{L}^{k}}) \int _{M} c_{\mathcal{L}^{k},z}\,d\mu(z) \in \operatorname{End}(\mathcal{H}_{\mathcal{L}^{k}}) \]
where $\mu$ is the Fubini-Study volume measure on $M$ normalized to have total mass 1.
If each $A_{k}$ is close to the identity on $\mathcal{H}_{\mathcal{L}^{k}}$ in the trace norm and the error decreases to $0$ as $k \to \infty$ then we would have an approximate coherent state POVM.

This statement in the operator norm follows from \cite[Theorem 1.8]{Finski_2022_2}, which considers the operator norms of the restriction map $\Gamma(\mathbb{P}(V), \mathcal{O}(1)^{k}) \to \Gamma(M, \mathcal{L}^{k})$ where the inner products are given by integration with respect to some Riemannian volume forms on $\mathbb{P}(V)$ and $M$.
In particular, we can normalize by $\dim \mathcal{H}_{\mathcal{L}^{k}}$ and $\dim \Gamma(\mathcal{O}(1)^{k})$, which by the Hirzbruch--Grothendieck--Riemann--Roch theorem (or Atiyah--Singer index theorem applied to the $\bar{\partial}$ operator) is asymptotically $\operatorname{vol}(M)k^{\dim M}$. Since the measure is normalized by this dimension, we obtain the desired statement with the correct normalization.
One can also show convergence in the trace norm by taking traces of the corresponding Schwartz kernels of the operators, which converge by \cite[Theorem 5.1]{Finski_2022_2}. In particular, combining this argument with the above Proposition~\ref{prop:approx-povm} gives an asymptotic form of the de Finetti theorem. The proof of \cite{Finski_2022_2} relies on the asymptotic expansion of the Bergman kernel for polarized K\"{a}hler manifolds, which becomes dramatically simpler when $M$ and the metric on $\mathcal{L}$ are invariant under a transitive group action.

To obtain error bounds on the HSoS hierarchy, one needs to bound the ratio $P_{n - k}/P_{n}$ of the densities defining the coherent state POVMs, in addition to incurring an additive error from the incompleteness of the POVMs. In particular, by applying this method one incurs an additive error of $\epsilon_{n} \to 0$ as $n \to \infty$, and then a multiplicative error of $Ck/n$ as in the standard quantum de Finetti theorem but with a constant $C$ depending on the manifold $M$ instead of the usual factor of $d$. It is possible to explicitly calculate bounds on $\epsilon_{n}$ and $C$ in terms of polynomials in the curvature tensor of the induced Fubini-Study metric on $M$, by applying results from \cite{Finski_2022_2} on the Taylor expansion of the Toeplitz kernel of the sequence of the operators $A_{k}$. In future work the author will give a rigorous proof that the POVM densities $P_{n}$, for sufficiently large $n$, can be taken to be smooth perturbations of the K\"{a}hler volume measure of $M$, which would give a purely multiplicative error bound without any additional additive error.

More qualitatively, a Hermitian Positivstellensatz is known in general. The proofs by Quillen \cite{Quillen_1968} (in the case $M = \mathbb{C}P^{d}$) and Catlin--D'Angelo \cite{Catlin_DAngelo_1999} (for an arbitrary $M$ and $\mathcal{L}$) can be interpreted as considering the direct sum of analog of the P- and Q-quantizations of a Hermitian polynomial function $f$ on the direct sum $\bigoplus _{k \geq 0} \mathcal{H}_{\mathcal{L}^{k}}$ where the inner products on each $\mathcal{H}_{\mathcal{L}^{k}}$ are rescaled by a constant $c_{k}$ (which differs between the two proofs).
The asymptotic expansion of the Bergman kernel is used again to show that the difference between the direct sums of the P- and Q-quantizations is a compact operator by realizing the direct sum as either the Segal-Bargmann space (in Quillen's proof) or the space of holomorphic functions on the unit ball (in the Catlin-D'Angelo's proof).
Since the P-quantization is manifestly PSD whenever $f$ is nonnegative this implies that the Q-quantization is positive operator for $k$ sufficiently large. However, the analogs of the P- and Q-quantizations would be defined with respect to the volume measure of the K\"{a}hler metric on $M$, which is different from the quantization rules we have defined here (but coincides when there is a transitive group action fixing the metric).
For a de Finetti theorem one needs the tensor product property of the coherent states and so this essentially fixes the quantization rule to be the one we have defined here.

For the analog of the Kikuchi hierarchy we describe in Section~\ref{sec:kikuchi} convergence at a rate of $O(\hbar) = O(1/k)$ has been shown in \cite{Le_Floch_2018}, except that the inner products on $\mathcal{H}_{\mathcal{L}^{k}}$ are given by the operators $A_{k}$ we defined and not by restrictions of the inner products on $\Gamma(\mathcal{O}(1)^{k})$.
This result again relies on the asymptotic expansion of the Bergman kernel, and relies also on using the volume measure of the K\"{a}hler metric. From the proof of Proposition~\ref{prop:polys-dense} we can deduce a similar result using the definition of the P-quantization as in Definition~\ref{def:Pquant}, showing that as $k \to \infty$ the hierarchy of lower bounds on $\sup _{z \in M} f(z)$ (given by the spectral norms of the P-quantizations of $f$) converges asymptotically to the true value, although defining the P-quantizations in this case again requires using the volume measure on $M$.

\begin{remark}
Formulas for the Bergman kernel are relatively simple when there is a transitive group action fixing the metric on $\mathcal{L}$, it seems that in general not much is known beyond the existence of an asymptotic expansion and some low-order terms. Some explicit formulas for some particular cases of domains in $\mathbb{C}^{n}$ are obtained in \cite{Huo_2016}.
\end{remark}

\begin{remark}
It seems that proving asymptotic convergence of the HSoS hierarchy defined by restricting the inner products on the symmetric subspaces may not be possible without some of the differential and complex geometry tools described above.
If $M \hookrightarrow \mathbb{P}(V)$ is a smooth projective variety embedded in projective space, one can compose with the Veronese embedding to get $M \hookrightarrow \mathbb{P}(\operatorname{Sym}^{l}(V))$.
One can then take an arbitrary inner product on $\operatorname{Sym}^{l}(V)$ and construct the HSoS hierarchy by restricting the inner products on $\operatorname{Sym}^{k}(\operatorname{Sym}^{l}(V)) \subseteq (\operatorname{Sym}^{l}(V))^{\otimes k}$ in the same way, obtaining pullbacks of the Fubini-Study metric on $M$.
However, a theorem of Tian \cite{Tian_1990} and further results \cite{Zelditch_1998} show that pullbacks of Fubini-Study metrics to $M$ are dense in the space of all K\"{a}hler metrics on $M$ under the $C^{\infty}$ topology. Thus it seems that any asymptotic convergence result for the HSoS hierarchy is enough to conclude a statement about all K\"{a}hler metrics on $M$.
\end{remark}

\subsubsection{A Kikuchi-like hierarchy}\label{sec:kikuchi}
The Kikuchi hierarchy was first defined over the boolean hypercube in \cite{Wein_Alaoui_Moore_2019} as a spectral algorithm meant to be an alternative to the usual sum-of-squares hierarchy, which in general is an SDP.
In \cite{Hastings_2020} an adaptation of the Kikuchi hierarchy for optimization over spheres was proposed, also by giving a rule to convert a classical polynomial objective function into a quantum Hamiltonian. However, the quantization rule proposed in \cite{Hastings_2020} is the same as the quantization rule we described in Section~\ref{sec:sparse}, where all annihilation operators are to the right and all creation operators are to the left, and it was shown that this gives the Hermitian sum-of-squares hierarchy.
We will show here that exactly the opposite quantization rule, where all annihilation operators are to the left and all creation operators are to the right, gives a better analog of the Kikuchi hierarchy. In the rest of this section, we will assume $M = \mathbb{C}P^{d}$ and $\mathcal{L}$ is the anticanonical bundle over $M$.

To start, we give an alternative definition of the Kikuchi hierarchy. In \cite{Wein_Alaoui_Moore_2019}, the hierarchy was defined over the boolean hypercube $\{\pm 1\}^{n}$ as the linear map sending the monomial $x^{S}$ (where $S \subseteq \{0, 1\}^{n}$) to the matrix
\[ A_{k}(x^{S}) = [1_{\{T \Delta R = S\}}]_{R,T \in \binom{[n]}{k}} \]
whose rows and columns are indexed by subsets of $[n]$ of size $k$, and extended by linearity. In the Fourier basis, this can alternatively be rewritten as
\[ A_{k}(x^{S}) = \Pi_{\binom{[n]}{k}} M_{x^{S}} \Pi_{\binom{[n]}{k}} \]
where $\Pi_{\binom{[n]}{k}}$ is the projector onto the subspace of functions on $\{\pm 1\}^{n}$ which are linear combinations of monomials of degree exactly $k$, and $M_{x^{S}}$ is the operator which pointwise multiplies a boolean function by the polynomial function $x^{S}$. Since the eigenvalues of $M_{p}$ are exactly the values of $p$ evaluated on all points of the boolean hypercube, the spectral norms of the operators $A_{k}(p)$ give lower bounds on the maximum of $p$.

By analogy to the definition of $A_{k}$ in terms of projectors, we will consider a bihomogeneous Hermitian polynomial $p$ over $M$ of degree $k$ and define
\[ A_{k}(p) \coloneqq \Pi_{k} M_{p} \Pi_{k} \]
where $M_{p} : L^{2}(M, \mathcal{L}^{k}) \to L^{2}(M, \mathcal{L}^{k})$ the multiplication operator of $p$ acting on $L^{2}$ sections, and $\Pi_{k} : L^{2}(M, \mathcal{L}^{k}) \to \mathcal{H}_{\mathcal{L}^{k}}$ is the orthogonal projection onto the subspace of holomorphic sections.

We claim that $A_{k}(p)$ is in fact the P-quantization of $p$, that is
\[ A_{k}(p) = \int _{M} p(z)c_{\mathcal{L}^{k},z}\,d\mu(z). \]
Suppose that $s$ is a holomorphic section of $\mathcal{L}^{k}$. Then
\begin{align*}
\bra{s} \left(\int _{M} p(z)c_{\mathcal{L}^{k},z}\,d\mu(z)\right) \ket{s} &= \int _{M} p(z)\braket{s|c_{\mathcal{L}^{k},z}|s}\,d\mu(z) \\
&= \int _{M} p(z)|s(z)|^{2}\,d\mu(z)
\end{align*}
and
\begin{align*}
\braket{s|\Pi_{k}M_{p}\Pi_{k}|s} &= \langle \Pi_{k}s, M_{p} \Pi_{k} s \rangle_{L^{2}(M, \mathcal{L}^{k})} \\
&= \langle s, M_{p} s \rangle_{L^{2}(M, \mathcal{L}^{k})} \\
&= \int _{M} p(z)|s(z)|^{2}\,d\mu(z).
\end{align*}
Thus the quadratic forms defined by the two operators are equal, and hence the operators themselves are equal.

It is known \cite{Cahill_Glauber_1969} that the Glauber-Sudarshan P-quantization rule for the phase space $\mathbb{C}^{d}$ can equivalently be defined as
\[ p \mapsto \int _{M} p(z)\ket{z}\bra{z}\,d\mu(z) \]
or if $p$ is a polynomial, by quantizing such that all the creation operators are to the right and the annihilation operators are to the left. By the correspondence between coherent states for $\mathbb{C}^{d}$ and coherent states for $\mathbb{C}P^{d}$ discussed in Section~\ref{sec:sparse} and a similar argument as for the Husimi quantization rule, it follows that the P-quantization rule for $\mathbb{C}P^{d}$ can similarly be described using creation and annihilation operators, thus proving the statement asserted earlier.

\nocite{*}
\printbibliography

\end{document}